\documentclass[journal,10pt,twocolumn]{IEEEtran}

\usepackage[english]{babel}
\usepackage{comment}

\usepackage{amsmath}
\usepackage{amssymb}
\usepackage{mathtools}
\usepackage{mathbbol}

\usepackage{graphicx}
\usepackage{float}

\usepackage{sgame}

\usepackage[colorlinks=true, allcolors=blue]{hyperref}
\usepackage[capitalize]{cleveref}
\usepackage{cite}

\usepackage{xcolor}
\usepackage{enumerate}

\newtheorem{remark}{Remark}
\newtheorem{proposition}{Proposition}
\newtheorem{theorem}{Theorem}

\newtheorem{corollary}{Corollary}
\newtheorem{definition}{Definition}

\newtheorem{lemma}{Lemma}

\DeclareMathOperator{\Equaldef}{\overset{def}{=}}

\title{Learning to Coordinate over Networks with Bounded Rationality}
\author{Zhewei Wang, Emrah Akyol and Marcos M. Vasconcelos}

\begin{document}
\maketitle

\begin{abstract}
Network coordination games are widely used to model collaboration among interconnected agents, with applications across diverse domains including economics, robotics, and cyber-security. We consider networks of bounded-rational agents who interact through binary stag hunt games, a canonical game theoretic model for distributed collaborative tasks. Herein, the agents update their actions using logit response functions, yielding the well-known Log-Linear Learning (LLL) algorithm. While convergence of LLL to a risk-dominant Nash equilibrium of potential games requires unbounded rationality, we consider regimes in which rationality is strictly bounded.
We first show that the stationary probability of states corresponding to perfect coordination is monotone increasing in the rationality parameter~$\beta$. For $K$-regular networks, we prove that the stationary probability of a perfectly coordinated action profile is monotone in the connectivity degree $K$, and we provide an upper bound on the minimum rationality required to achieve a desired level of coordination. For irregular networks, we show that the stationary probability of perfectly coordinated action profiles increases with the number of edges in the graph.
To analyze these stationary distributions, we study Gibbs measures using a Gaussian approximation for the potential function when the admissible action profiles are uniformly distributed. We show that, for a large class of networks, the partition function of the Gibbs measure is well approximated by the moment generating function of Gaussian random variable. This approximation allows us to optimize degree distributions and establishes that the optimal network—i.e., the one that maximizes the stationary probability of coordinated action profiles—is $K$-regular. Consequently, our results indicate that networks of uniformly bounded-rational agents achieve the most reliable coordination when connectivity is evenly distributed among agents.
\end{abstract}

\section{Introduction}

One of the possible applications that calls for the deployment of a multi-agent system is when there is a collective task (or a job) whose difficulty exceeds the capabilities of any individual agent operating in the environment. In this situation, at least a subset of the agents in the system need to work together to perform the task, and such synergistic behavior requires {\it coordination}. As a foundational principle in robotics, economics, computer science and microbiology, achieving coordination is a desirable feature and as such has been studied from the point of view of many mathematical models, including game theory.

Coordination games are simultaneous-move games in which agents benefit from choosing the same action. Among these, the stag hunt game \cite{skyrms2004stag} captures the tension between a safe, low-reward action and a risky, high-reward action that requires cooperation. This simple model of incentives for collaborative interaction can be extended over a multi-agent network, where an agent interacts with a subset of all agents, called a neighborhood, leading to a much more complex and realistic setting suitable for designing modern engineering applications and analyzing socioeconomic phenomena.

In practice, learning agents may not best-respond perfectly due to cognitive limitations, computational constraints, or stochastic execution errors - a condition broadly referred to as \textit{bounded rationality} \cite{simon1955behavioral}. 
While network coordination games provide a rich mathematical framework, a system designer interested in orchestrating collective behavior, must contend with bounded rational agents.
This is the case when the agents in our model are humans in socioeconomic networks, their decisions are influenced by highly subjective factors inherent to the human condition \cite{jackson2008social}. For instance, in the stag hunt game, the choice between hunting a hare or a stag may vary significantly across individuals and may not always be rationalizable \cite{kim2025experimental,camerer1999experience,camerer2003behavioral}. Similarly, engineered agents such as robots or AI agents may not always be able to perform perfect optimization due to computational constraints or model hallucinations. In such cases a suboptimal solution must be implemented \cite{Vamvoudakis:2023}. Other times, even if an agent can optimize perfectly, they may fail to execute that particular action due to the stochastic nature of the environment. Therefore, bounded rationality is a limiting factor on the predictability of the system behavior \cite{tsiotras2021bounded}.

In this paper, we analyze the interplay of bounded rationality and connectivity in network coordination games. We focus on a binary stag hunt game in which agents decide whether to attempt a collective task of varying difficulty with the help of their neighbors. The underlying model assumes the agents play the game repeatedly, using a logit response dynamics with bounded rationality, seeking to learn to play a coordinated action profile in the network game. Due to the bounded rationality of the agents, there is a non-vanishing probability of miscoordination. We show that the probability of coordination with bounded rationality can be improved by increasing the connectivity of the network. This result is shown both for $K$-regular and for irregular networks. Then, we show that for a sufficiently large number of agents with a homogeneous level of bounded rationality, the networks that maximize the probability of coordination are $K$-regular (when one exists for the parameters of the game), or near $K$-regular. This set of results provides an important design principle for multi-agent network systems with homogeneous bounded rationality: for systems with a large number of agents, the designer can adjust the number of neighbors to achieve a prescribed level of coordination in the long run, even though the agents are responding to the actions of other agents imperfectly.

\subsection{Related Literature}

\subsubsection{Network Coordination Games}

Network games have been extensively studied as models of strategic interaction among agents whose payoffs depend on the actions of their neighbors in a graph. The framework of graphical games was introduced by Kearns et al. in \cite{kearns2001graphical}, which consisted of an undirected graph and a set of local payoff matrices for
multi-player games. Kakade et al. \cite{kakade2004graphical} showed how graph structure can be exploited for efficient computation of Nash equilibria. Since their introduction, a rich literature has been developed propelled by the popularization of social networks. Jackson and Zenou~\cite{jackson2015games} provide a comprehensive survey of games on networks, covering both complete-information and Bayesian settings. Coordination games on networks, in which agents benefit from aligning their actions with neighbors, have been studied in many different contexts~\cite{morris2000contagion,jackson2015games,montanari2010spread,young1993evolution}. Variants of the base model that incorporate the ability to respond to cyberattacks and external biases has been proposed in \cite{paarporn2020risk,paarporn2020impact,Arditti:2023}. A key insight from this literature is that the structure of the network (degree distribution and connectivity) plays a fundamental role in determining which equilibria are selected and how quickly agents converge to them. Our work contributes to this line of research by characterizing how network topology interacts with bounded rationality to affect coordination outcomes under stochastic learning dynamics.

\subsubsection{Models of Bounded Rationality}

Bounded rationality has a long history tied to the literature on behavioral economics, originating with the seminal work of Simon~\cite{simon1955behavioral}, who argued that human decision makers operate under cognitive and computational constraints and therefore, are unable of achieving perfectly rational behavior. Since  the work of Simon, many different models of bounded rationality have emerged. Prospect theory~\cite{kahneman1979prospect,rieger2008prospect,nadendla2018effects}, which models risk-sensitive decision making under uncertainty; level-$K$ thinking and cognitive hierarchy models~\cite{nagel1995unraveling,camerer2004cognitive,kokolakis2023bounded,ferreira2021risk,yoshida2008game}, which assume agents perform a limited number of strategic reasoning steps predicting sequences of best-responses to best-responses up to a certain level determined by the cognitive capacity of the agent; and the quantal response equilibrium (QRE)~\cite{mckelvey1995quantal,goeree2016qre,melo2022uniqueness}, which replaces exact best responses with a stochastic choice rules known as a quantal best response, generalizing the notion of a Nash equilibrium when the agents no longer respond optimally to the others decisions. The QRE framework is closely related to the discrete choice models of McFadden~\cite{mcfadden1973conditional} and naturally gives rise to the logit dynamics considered herein.

\subsubsection{Log-Linear Learning}

Our approach to bounded rationality is based on Log-Linear Learning (LLL), which is an interactive algorithm where the agents repeatedly play the game revising their actions given the actions played by their neighbors \cite{marden2012revisiting}. In LLL, the agents respond suboptimaly using a logit kernel, which is similar to a quantal-best response, but subtly different in that the agents respond to their neighbors actions and not to their mixed strategies. LLL was introduced by Blume~\cite{blume1993statistical,blume1995statistical} and has since been extensively studied in the context of potential games~\cite{alos2010logit,marden2012revisiting,akbar2022robustness,Maddux:2026}. One can think of LLL as a form of stochastic distributed optimization algorithm, where the global function being maximized is the game's potential function. In an analogy with {\it annealing} \cite{kirkpatrick1983optimization}, the literature on LLL primarily focuses on the asymptotics when the rationality (inverse temperature) grows without bound. In our model, however, rationality is kept bounded and we use the network as a means to compensate for such limitation.

\subsection{Our Contributions}

The main contributions of this work are as follows.

\begin{itemize}

    \item We define a binary network coordination stag hunt game, and show that when the graph is undirected, this is a potential game.

    \item Under LLL with homogeneous bounded rationality, we show that connectivity improves the probability that the agents will asymptotically play one of the two pure NE of the game.

    \item When the network is $K$-regular, we show that the minimum rationality to achieve coordination with high probability is inversely proportional to the connectivity.

    \item Using a Martingale Central Limit theorem, we show that under mild technical conditions the potential function evaluated at uniformly distributed action profiles converges in distribution to a Gaussian random variable, whose variance only depends on the degree distribution, thereby enabling optimization via Majorization theory.

    \item We show that for a sufficiently large number of agents, regular graphs maximize the probability that homogeneous bounded-rational agents converge to a NE.

\end{itemize}

Preliminary versions of some of the results in this paper have been presented in \cite{zhang2024rationality} and \cite{zhang2024role}. The present work significantly extends the scope of those contributions by providing complete proofs, extending the analysis to irregular graphs, and establishing the optimality of regular networks in the large-network regime and small rationality regimes.

\subsection{Notation}

We use $[N]\Equaldef\{1,2,\ldots, N\}$ to denote the set of agents. The symbols $\mathbb{0}$ and $\mathbb{1}$ denote the all-zeros and all-ones vectors in $\mathbb{R}^N$, respectively. For a vector $a\in\{0,1\}^N$, we denote by $\|a\|_1 = \sum_{i=1}^{N}a_i$ its $\ell^1$ norm and by $\|a\|_2 = (\sum_{i=1}^N a_i^2)^{1/2}$ its $\ell^2$ norm; since $a$ is binary, $\|a\|_1 = \|a\|_2^2$. Graphs are denoted by $\mathcal{G}=([N],\mathcal{E})$, where $\mathcal{E}\subseteq [N]\times [N]$ is the edge set. The adjacency matrix of $\mathcal{G}$ is denoted $\mathbf{A}\in\{0,1\}^{N\times N}$, and the graph Laplacian is $\mathbf{L}\Equaldef \mathbf{D}-\mathbf{A}$, where $\mathbf{D}\Equaldef \mathrm{diag}(\mathbf{A}\mathbb{1})$ is the degree matrix. The neighborhood of agent $i$ is $\mathcal{N}_i \Equaldef \{j\in[N]\mid (i,j)\in\mathcal{E}\}$, and the degree of agent $i$ is $d_i\Equaldef |\mathcal{N}_i|$. We write $a_{-i}$ for the vector of actions of all agents except agent $i$, and $a_{\mathcal{N}_i}$ for the sub-vector of actions of agent $i$'s neighbors.  For a matrix $\mathbf{M}$, $\|\mathbf{M}\|_F$ denotes its Frobenius norm and $\|\mathbf{M}\|_2$ its spectral (operator) norm. The notation $\xrightarrow{D}$ denotes convergence in distribution and $\xrightarrow{P}$ denotes convergence in probability.

\subsection{Organization}

The remainder of the paper is organized as follows. \Cref{sec:setup} introduces the problem setup, the binary stag hunt coordination game and its extension to networks. \Cref{sec:potential} establishes that the network coordination game is an exact potential game and characterizes the maximizers of the potential function. \Cref{sec:tradeoff} analyzes the trade-off between bounded rationality and connectivity under Log-Linear Learning for $K$-regular graphs, by proving monotonicity of stationary probability of coordinated action profiles in both $\beta$ and $K$, and deriving an upper bound on the minimum rationality required for coordination. \Cref{sec:irregular} extends the analysis to irregular graphs, showing that coordination probability increases with edge connectivity. \Cref{sec:gaussian} addresses the optimal network design problem: we show that the partition function of the Gibbs distribution proportional to a moment generating function, and use this connection to prove the optimality of regular graphs in two regimes: small $\beta$ (via Taylor series expansion) and large $N$ (via a Gaussian approximation). When $\beta$ and $N$ are moderate, we show that regular graphs maximize a nontrivial spectral lower bound on the stationary probability of coordinated action profiles. Finally, \cref{sec:conclusions} concludes the paper and discusses directions for future work.

\section{Problem Setup}\label{sec:setup}

Consider a binary action networked coordination game with $N$ agents. Let $[N]\Equaldef\{1,2,\ldots, N\}$ denote the set of agents, whose interactions are described by an undirected and connected graph $\mathcal{G}\Equaldef([N],\mathcal{E})$. Each agent $i \in [N]$ has a binary action space $\mathcal{A}_i = \{0,1\}$.
Two nodes $i,j \in[N]$ are  connected if $(i,j)\in \mathcal{E}$. The set of neighbors of agent $i$ is denoted by $\mathcal{N}_{i} \Equaldef \{j\in [N] \mid (i,j)\in \mathcal{E}\}$. The number of neighbors of agent $i$ is denoted by $|\mathcal{N}_i|$. We assume there are no self loops, i.e., $(i,i) \notin \mathcal{E},$ $i\in[N]$.

Let $(i,j)\in\mathcal{E}$, and suppose that $a_i,a_j\in\{0,1\}$ are the actions played by agents $i$ and $j$, respectively. Let $\theta \in \mathbb{R}$. The following bimatrix game specifies the payoffs for the pairwise interaction between agents $i$ and $j$.

\begin{figure}[h!]\hspace*{\fill}%
\centering
\begin{game}{2}{2}[$a_i$][$a_j$]
     & $1$ & $0$             \\
 $1$ & $\Big( 1-\theta, 1-\theta$\Big) & $\Big(-\theta,0\Big)$ \\
 $0$ & $\Big(0,-\theta\Big)$ & $\Big(0,0\Big)$  \\
\end{game}\hspace*{\fill}%
\caption{A coordination game with parameter $\theta$ between two players.}
\label{bimatrix}
\end{figure}

\vspace{5pt}

\begin{remark}[Payoff interpretation]
The payoff structure of the bimatrix game in \cref{bimatrix} corresponds to a \textit{stag hunt} coordination game \cite{skyrms2004stag} between two agents $i$ and $j$. Notice the payoff matrix depends on the {\it task difficulty} $\theta \in \mathbb{R}$.
\end{remark}

A (binary) stag hunt coordination game between two agents is characterized by the existence of two pure strategy Nash equilibria. The following result establishes the range of values of $\theta$ for which the game in \cref{bimatrix} corresponds to a coordination game.

\vspace{5pt}

\begin{proposition}\label{pure_strategy_equilibria}
    Consider the bimatrix game in \cref{bimatrix}, and let $\mathcal{S}_{ij}$ denote its set of pure-strategy Nash equilibria. Then,
\begin{equation}
\mathcal{S}_{ij} =
\begin{cases}
\big\{ (0,0) \big\} & \text{if} \ \  \theta > 1 \\
\big\{ (0,0),(1,1) \big\} & \text{if} \ \ 0 \leq \theta \leq 1 \\
\big\{ (1,1) \big\} & \text{if} \ \ \theta < 0.
\end{cases}
\end{equation}
\end{proposition}

\vspace{5pt}

\begin{IEEEproof}
    The proof can be obtained by inspection using the definition of a Nash equilibrium \cite{fudenberg1991game}.
\end{IEEEproof}

\vspace{5pt}

\subsection{Coordination games over networks}

We study a network coordination game with $N$ agents, where agent $i$ plays the same action with all of its neighbors $j\in\mathcal{N}_i$. Let $V_i:\{0,1\}^2\rightarrow \mathbb{R}$ be defined as
\begin{equation}
V_i(a_i,a_j) \Equaldef a_i\Big(
a_j-\theta\Big).
\end{equation}

In a network game, the payoff that one player receives is the sum of all the payoffs of the bimatrix games $V_{i}(a_{i}, a_{j})$  played with each one of its neighbors. Therefore for the $i$-th player, the utility is determined as follows
\begin{equation}
    U_{i}(a_{i},a_{-i})\Equaldef  \sum\limits_{j\in \mathcal{N}_{i}} V_{i}(a_{i},a_{j}).
\end{equation}
The payoff of the $i$-th agent in our game is
\begin{equation}\label{eq:payoff}
U_{i}(a_{i},a_{-i}) = a_i\Big(\sum_{j\in \mathcal{N}_i}a_j-\theta|\mathcal{N}_i|\Big).
\end{equation}

\section{Potential Network Coordination Games}\label{sec:potential}

The network stag hunt coordination game considered herein is always an {\it exact potential game} regardless of the graph structure.

\subsection{Potential games}

\begin{definition}\label{exact_potential}
Let $\mathcal{A}_i$ denote the action set of the $i$-th agent in a game with payoff functions $U_i(a_i,a_{-i})$, $i\in[N]$. Let $\mathcal{A} = \mathcal{A}_1\times \cdots \times \mathcal{A}_n$. A game is an exact potential game if there is a  \textit{potential function} {$\Phi$}:
 {$\mathcal{A}\rightarrow \mathbb{R}$} such that
\begin{equation}\label{eq:exact_potential1}
    U_{i}(a'_{i},a_{-i})- U_{i}(a''_{i},a_{-i}) = \Phi(a'_{i},a_{-i})- \Phi(a''_{i},a_{-i}),
\end{equation}
for all  $a_{i}',a_{i}''\in \mathcal{A}_{i}$, $a_{-i}\in \mathcal{A}_{-i}$,  $i \in [N]$.
\end{definition}

\vspace{5pt}

\begin{theorem}\label{theorem_EP}
Let $\mathcal{G}=([N],\mathcal{E})$ be an undirected and connected graph. Consider a networked
coordination game defined by the payoff in \eqref{eq:payoff} indexed by the parameter $\theta$.
The game is an exact potential game for any {$\theta$}.
\end{theorem}
\vspace{5pt}
\begin{IEEEproof}
The proof is in Appendix \ref{sec:proofs}.
\end{IEEEproof}

\vspace{5pt}

We have established that this networked coordination game always an exact potential game. In the next proposition, we  obtain a closed form expression for its potential function.

\vspace{5pt}

\begin{proposition}\label{pot}
Let $\mathbf{A}\in\{0,1\}^{N\times N}$ be the adjacency matrix of a graph $\mathcal{G}$. The exact potential function for the network coordination game defined over $\mathcal{G}$ is given by $\Phi_{\mathbf{A}}(a)$ defined as
\begin{equation}\label{potential_function}
\Phi_\mathbf{A}(a) \Equaldef \frac{1}{2} a^\mathsf{T}\mathbf{A} a -  \theta  \mathbb{1}^{\mathsf{T}}\mathbf{A}a + \frac{\theta}{2}\mathbb{1}^\top \mathbf{A}\mathbb{1} ,
\end{equation}
where $\theta$ is the task difficulty and $a \in \{0,1\}^N$ is the action profile.
\end{proposition}

\vspace{5pt}

\begin{IEEEproof}
The proof follows from equations \eqref{potential_f} and \eqref{eq:potential_two_agents} by expanding the sums and using the symmetry of $\mathbf{A}$.
\end{IEEEproof}

\vspace{5pt}

A seminal result by Monderer and Shapley~\cite{monderer1996potential} establishes that, 
in an exact potential game, a strategy profile is a pure-strategy Nash equilibrium if and 
only if it is a local maximizer of the potential function. Consequently, identifying all 
pure-strategy Nash equilibria of the game is equivalent to finding all local maximizers 
of $\Phi$. We will show that when the graph is connected, the potential function is maximized when every agent in the system plays the same action. We proceed with the characterization of the set of optimal solutions for the following optimization problem
\begin{equation}
\begin{aligned}\label{OriginalProblem}
& \underset{a \in {\{ 0,1 \}}^{N}}{\mathrm{maximize}}
& &\frac{1}{2} a^\mathsf{T}\mathbf{A} a -  \theta  \mathbb{1}^{\mathsf{T}}\mathbf{A}a \Equaldef f_0(a).
\end{aligned}
\end{equation}

\vspace{5pt}

\begin{theorem}\label{0-1property}
Consider a connected undirected graph $\mathcal{G}$, with an adjacency matrix $\mathbf{A}$. Let $\mathcal{S}^\star_{\mathcal{G}}$ denote the set of maximizers of the potential $\Phi_{\mathbf{A}}(a)$ for the network coordination game defined over $\mathcal{G}$. Then,
\begin{equation}
\mathcal{S}^\star_{\mathcal{G}}\subseteq \{\mathbb{0},\mathbb{1} \}.
\end{equation}
\end{theorem}

\vspace{5pt}

\begin{IEEEproof}
 Rewriting the objective function in \eqref{OriginalProblem} in terms of the graph Laplacian\footnote{The graph Laplacian is defined as $\mathbf{L}\Equaldef \mathbf{D}-\mathbf{A}$ such that $\mathbf{D}\Equaldef \mathrm{diag}(\mathbf{A}\mathbb{1})$.}, we obtain
\begin{equation}
f_0(a) = \Big(\frac{1}{2}-\theta \Big)d^\top a - \frac{1}{2}a^\top\mathbf{L}a,
\end{equation}
where $d\Equaldef \mathbf{A}\mathbb{1}$ denotes the graph's degree sequence. Since $\mathbf{L}$ is always a positive semidefinite matrix \cite{FB-LNS}, if the graph is connected, the following holds
\begin{equation}
a^\top \mathbf{L} a \geq 0, \ \ a\in\{0,1\}^N,
\end{equation}
with equality if and only if $a\in \{\mathbb{0},\mathbb{1}\}.$ Therefore,
\begin{equation} \label{eq:upper_bound}
\max_{a\in\{0,1\}^N} f_0(a) \leq \max_{a\in\{0,1\}^N} \Big(\frac{1}{2}-\theta \Big)d^\top a.
\end{equation}

Since the function on the right hand side of \eqref{eq:upper_bound} is linear in $a$, and $d_i\geq 0$ for all $i\in[N]$, it is either increasing or decreasing depending on $\theta$, which implies that
\begin{equation}\label{eq:optimal}
a^\star = \begin{cases}
\mathbb{0} & \text{if} \ \ \theta > \frac{1}{2} \\  
\mathbb{0} \ \text{or} \ \mathbb{1} \ & \text{if} \ \ \theta = \frac{1}{2} \\
\mathbb{1} & \text{if} \ \ \theta < \frac{1}{2}.
\end{cases}
\end{equation}
Therefore,
$\mathcal{S}_{\mathcal{G}}^\star = \arg \max \big\{f_0(a) \mid a \in \{0,1\}^N\big\} \subseteq \{\mathbb{0},\mathbb{1}\}.$
\end{IEEEproof}

\section{Trade-off Between Rationality and Connectivity}\label{sec:tradeoff}

The correspondence between maximizers of the potential function and pure-strategy Nash equilibria establishes a link between optimization and the rational behavior of agents playing our network coordination game. Moreover, the equilibrium selected through interactive game play can be justified using the Log Linear Learning framework \cite{blume1993statistical,marden2012revisiting,alos2010logit}. When bounded-rational agents gradually increase their rationality over time the learning dynamics converge to the risk-dominant equilibrium, which coincides with the global maximizer of the potential function. In the limit of infinite rationality, the network connectivity does not affect the induced Markov chain induced by LLL in the action space. However, it affects its convergence rate \cite{montanari2010spread,arieli2020speed}. In this section we will establish that the connectivity and rationality have a non-trivial interplay in the bounded rationality regime with respect to the stationary probability of coordinated action states. In the next subsection, we describe the LLL framework.

\subsection{Log-Linear Learning with Bounded Rationality}

Suppose that the agents in the network coordination game interact asynchronously over time as follows. At time $t=0$, agent $i$ picks an action $a_i(0) \in \{0,1\}$, $i\in[N]$. At all subsequent times $t>0$, an agent is randomly selected with uniform probability, observes noiselessly the actions of its neighbors at the previous time, $a_{\mathcal{N}_i}(t-1) \Equaldef \{a_j(t-1) \mid j\in \mathcal{N}_i\}$, and updates its action according to a logit stochastic kernel defined as
\begin{equation}
 \mathbb{P}\big(A_i(t) = a_i \mid A_{\mathcal{N}_i}(t-1) = a_{\mathcal{N}_i} \big)=\sigma_i(a_i,\beta \mid a_{\mathcal{N}_i}),
\end{equation}
where
\begin{equation}\label{eq:LLL}
\sigma_i(a_i,\beta \mid a_{\mathcal{N}_i})   \Equaldef \frac{e^{\beta U_{i}\big(a_{i},a_{\mathcal{N}_i}\big)}}{\sum_{a'_i\in \{0,1\}}e^{\beta U_{i}\big(a'_{i},a_{\mathcal{N}_i}\big)}}, \  a_i\in \{0,1\}.
\end{equation}

In behavioral economics, the logit kernel is used to model discrete choice under bounded rationality \cite{mcfadden1984econometric,train2009discrete,sandomirskiy2023independence,Sandomirskiy:2025,matejka2015rational}.
The parameter $\beta$ captures the agent's level of rationality, varying between random behavior ($\beta = 0$) and deterministic best-response behavior ($\beta \to \infty$).

The logit kernel defines a Markov chain with a state space $\mathcal{S} = \{0,1\}^N$ corresponding to all possible strategy profiles $a\in \mathcal{S}$ for the network coordination game. For exact potential games,
this Markov chain has a unique stationary distribution given by the Gibbs--Boltzmann distribution \cite{blume1993statistical,marden2012revisiting,montanari2010spread}. In particular, for our network coordination game defined over a graph with adjacency matrix $\mathbf{A}$, the stationary distribution  $\mu_\mathbf{A}: \mathcal{S} \rightarrow [0,1]$ is given by
\begin{equation}\label{eq:gibbs}
\mu_{\mathbf{A}}(a \mid \beta) \Equaldef
\frac{e^{\beta \Phi_{\mathbf{A}}(a)}}
{\sum_{a' \in \mathcal{S}} e^{\beta \Phi_{\mathbf{A}}(a')}} ,
\end{equation}
where $\Phi_{\mathbf{A}}$ is the potential function in \eqref{potential_function}. 

The existing analysis of LLL shows that as $\beta \to \infty$, the probability mass concentrates on the risk-dominant pure strategy NE, i.e., the maximizers of the potential function, which means that the only stochastically stable states of the Markov chain are the ones in $\mathcal{S}_{\mathcal{G}}^\star$. However, we are interested in analyzing the bounded rationality regime, $\beta<\infty$. In this case, the probability of any state distributed according to the Gibbs--Boltzmann distribution evaluated at $a^\star \in \mathcal{S}^\star_{\mathcal{G}}$ is bounded away from one.

In this section we are interested in the minimum value of $\beta$ such that the agents coordinate on one of the states in $\mathcal{S}_{\mathcal{G}}^\star$ with high probability. For $\delta \in (0,1)$ and $\theta \in[0,1]$, for a connected undirected graph with adjacency matrix $\mathbf{A}$ define
\begin{equation}\label{Betamin1}
\beta_{\mathbf{A}}^{\min}(\delta) \Equaldef \min \Big\{ \beta \mid \mu_{\mathbf{A}}(a^\star \mid \beta) \geq 1-\delta \Big\},
\end{equation}
where $a^\star$ is a maximizer of $\Phi_{\mathbf{A}}$ given by \eqref{eq:optimal}.

\subsection{Regular graphs}

The class of $K$-regular graphs is characterized by nodes that each have a constant number of neighbors, i.e., $|\mathcal{N}_i| = K$ for all $i \in [N]$ \cite{newman2018networks}. Restricting our analysis to $K$-regular graphs allows us to examine how $\beta_{\mathbf{A}}^{\min}(\delta)$ varies as a function of the connectivity parameter $K$.
Before discussing the interplay between rationality and connectivity in $K$-regular graphs, it is important to note that multiple non-isomorphic regular graphs may share the same degree $K$. These graphs cannot, in general, be related by a similarity transformation of their adjacency matrices. For instance, a bipartite and a non-bipartite regular graphs with the same degree are not isomorphic. Nevertheless, in what follows we construct a sequence of graphs $\{\mathbf{A}_K\}$ with increasing degree $K$, such that all graphs within the same isomorphism class yield the same value of $\beta_{\mathbf{A}_K}^{\min}(\delta)$.

Applying a similarity transformation is equivalent to re-assigning indices to agents. Although for a specific action profile $a \in \{0,1\}^N \backslash \{a^\star\}$, the corresponding potential value can be different on two isomorphic graphs $\mathcal{G}_1$ and $\mathcal{G}_2$ with the same $K$ and $N$, there exists a unique $\tilde{a} \in \{0,1\}^N$ such that $\Phi_{\mathcal{G}_1}(a) = \Phi_{\mathcal{G}_2}(\tilde{a})$. Such $\tilde{a}$ can be derived by applying the same similarity transformation on $a$. Therefore, when computing the exact value of $\mu_{\mathbf{A}}(a\mid\beta)$ for a specific $a \neq a^\star$, we must specify and fix a graph $\mathcal{G}$. Nevertheless, $\Phi_{\mathbf{A}}(a^\star)$ remains constant for all isomorphic graphs with the same degree and so does $\mu_{\mathbf{A}}(a^\star \mid \beta)$. This is further discussed in the proof of our next theorem.

\vspace{5pt}

The following lemma from graph theory characterizes the conditions under which a regular graph exists.

\vspace{5pt}

\begin{lemma}[\cite{diestel2017graph}]\label{graph_new}
    A simple $K$-regular graph $\mathcal{G}_K$ with $N$ vertices of degree $K$ exists if and only if $K \in \{0,\dots,N-1\}$ and $NK$ is even.
\end{lemma}

\vspace{5pt}

\begin{lemma}\label{G_K_new}
    Let $\mathbf{A}_K$ be the adjacency matrix of a connected $K$-regular graph $\mathcal{G}_K$. The following statements hold:

\begin{enumerate}[(a)]

\item If $N$ is even and $K<N-1$, then $\mathcal{G}_{K+1}$ always exists. Moreover, the adjacency matrix of a regular graph $\mathcal{G}_{K+1}$ can be constructed as follows: there exists
a symmetric permutation matrix $\Pi_1$, and a permutation matrix $\Pi_2$ such that
\begin{equation}
\mathbf{A}_{K+1} = \Pi_2(\mathbf{A}_{K}+\Pi_1)\Pi^\top_2.
\end{equation}

\item If $N$ is odd, $K$ is even and $K<N-2$, then $\mathcal{G}_{K+1}$ does not exist. However, $\mathcal{G}_{K+2}$ exists, and its adjacency matrix can be constructed as follows: there exist two distinct symmetric permutation matrices $\Pi_1, \Pi_2$ and a permutation matrix $\Pi_3$ such that
\begin{equation}
\mathbf{A}_{K+2} = \Pi_3(\mathbf{A}_{K}+\Pi_1+\Pi_2)\Pi_3^\top.
\end{equation}

 \end{enumerate}
 \end{lemma}

\vspace{5pt}

\begin{IEEEproof}
We start with part~(a). Since $N$ is even, $NK$ is even for any $K$. By \cref{graph_new}, $\mathcal{G}_K$ exists for every $K\in\{0,\ldots,N-1\}$. We construct $\mathcal{G}_{K+1}$ from $\mathcal{G}_K$ by adding a \textit{perfect matching} \cite{diestel2017graph}. Let $\bar{\mathcal{G}}_K$ denote the complement of $\mathcal{G}_K$. Since $\mathcal{G}_K$ is $K$-regular, $\bar{\mathcal{G}}_K$ is $(N-1-K)$-regular with $N-1-K\geq 1$. By the \textit{handshaking lemma} \cite{diestel2017graph}, $\bar{\mathcal{G}}_K$ has at least $N/2$ edges, and since it is regular of degree at least $1$ on an even number of vertices, it contains a perfect matching $\mathcal{M}$. Let $\Pi_1$ be the permutation matrix associated with the matching $\mathcal{M}$.

Adding $\mathcal{M}$ to $\mathcal{G}_K$ yields a $(K+1)$-regular graph whose adjacency matrix is $\mathbf{A}_K + \Pi_1$. The permutation matrix $\Pi_2$ accounts for a possible relabeling of the vertices.

For part~(b), when $N$ is odd and $K$ is even, $NK$ is even so $\mathcal{G}_K$ exists. However, $(K+1)N$ is odd, so by \cref{graph_new}, $\mathcal{G}_{K+1}$ does not exist. Since $(K+2)N$ is even, $\mathcal{G}_{K+2}$ exists. We construct it by adding two disjoint perfect matchings (symmetric permutation matrices $\Pi_1$ and $\Pi_2$) from the complement graph, with $\Pi_3$ introduced for node relabeling.
\end{IEEEproof}

\vspace{5pt}

The following lemma provides an upper bound on the binary quadratic form $a^{\mathsf{T}}\mathbf{A}_K a$.

\vspace{5pt}

\begin{lemma}\label{lem:bound_quadratic_form}
Let $\mathbf{A}_K$ be the adjacency matrix of a connected $K$-regular graph. Let $a \in \{0,1\}^N$ be such that $\|a\|_1 = m$. The following inequality holds
\begin{equation}\label{Norm}
a^{\top}\mathbf{A}_K a \leq mK , \ \ a \in \{0,1\}^N.
\end{equation}
\end{lemma}

\vspace{5pt}

\begin{IEEEproof}
Let $\|\mathbf{A}_K\|_2$ denote the $\ell^2$ induced operator norm\footnote {The $\ell^2$ induced operator norm of $\mathbf{A}_K$ is $\|\mathbf{A}_K\|_2 \Equaldef \sup_{x \neq \mathbb{0}}{\frac{\|\mathbf{A}_Kx\|_2}{\|x\|_2}}$.} of $\mathbf{A}_K$. It is well known that $\|\mathbf{A}_K\|_2$ is the largest singular value of $\mathbf{A}_K$, which in this case is $K$. Since operator norms are consistent with the vector norm inducing them, we have
\begin{equation}
    \|\mathbf{A}_Ka\|_2 \leq \|\mathbf{A}_K\|_2\|a\|_2, \ \ a \in \{0,1\}^N.
\end{equation}
Using the Cauchy--Schwarz inequality on $a^{\mathsf{T}}\mathbf{A}_K a$, we obtain
\begin{multline}
a^{\mathsf{T}}\mathbf{A}_K a \leq \|a\|_2  \|\mathbf{A}_Ka\|_2  \leq \|a\|^2_2\|\mathbf{A}_K\|_2 \\ = \|a\|_1\|\mathbf{A}_K\|_2= mK.
\end{multline}
\end{IEEEproof}

\vspace{5pt}

 Intuitively, $a^{\mathsf{T}}\mathbf{A}_K a$ measures the total interaction among the $m$ active nodes selected by the binary vector $a$. Since each node contributes at most $K$ connections in a $K$-regular graph, the bound $a^{\mathsf{T}}\mathbf{A}_K a \le mK$ follows from $\|\mathbf{A}_K\|_2 = K$.

From this point on, for simplicity, we  ignore the constant term in our potential function $\Phi_{\mathbf{A}_K}$ and use the following expression instead
\begin{equation}\label{PotentialFunction}
 \hat{\Phi}_{\mathbf{A}_K}(a) \Equaldef  \frac{1}{2}a^{\mathsf{T}}\mathbf{A}_K a -K\theta \sum_{i\in [N]}a_{i}.
\end{equation}

\begin{theorem}\label{thm:interplay}
Consider a network stag hunt coordination game defined over a connected $K$-regular graph $\mathcal{G}_K$ with adjacency matrix $\mathbf{A}_K$ and payoffs given by~\eqref{eq:payoff}.
Let the agents update their actions according to LLL with a rationality parameter $\beta \ge 0$.
Define
\begin{equation}
g(\beta,K) \Equaldef \mu_{\mathbf{A}_K}(a^\star \mid \beta),
\end{equation}
where $a^\star \in \{\mathbb{0},\mathbb{1}\}$  denotes a maximizer of the potential function $\Phi_{\mathbf{A}_K}$ and $\mu_{\mathbf{A}_K}(a^\star \mid \beta)$ is given by \eqref{eq:gibbs}.
Then, $g$ is strictly increasing in $\beta$ and monotone increasing in $K$. That is
\begin{enumerate}
    \item $g(\beta, K) < g(\beta, K+1)$ for even $N$;
    \item $g(\beta, K) < g(\beta, K+2)$ for odd $N$.
\end{enumerate}
\end{theorem}

\vspace{5pt}

\begin{IEEEproof}
First, we prove the monotonicity with respect to $\beta$. Computing the derivative of $g$ with respect to $\beta$, we obtain the following equivalence: $\frac{\partial g}{\partial \beta} > 0$ if and only if
\begin{multline}\label{eq:derivative_numerator1}
 \hat{\Phi}_{\mathbf{A}_K}(a^\star)e^{\beta \hat{\Phi}(a^\star)} \! \! \! \! \! \! \!\sum_{a'\in\{0,1\}^N}e^{\beta\hat{\Phi}_{\mathbf{A}_K}(a')} > \\  e^{\beta \hat{\Phi}_{\mathbf{A}_K}(a^\star)} \! \! \! \! \! \! \!  \sum_{a'\in\{0,1\}^N} \hat{\Phi}_{\mathbf{A}_K}(a')e^{\beta\hat{\Phi}(a')}.
\end{multline}
Since $e^{\beta\hat{\Phi}_{\mathbf{A}_K}(a^\star)}>0$, the condition in \eqref{eq:derivative_numerator1} becomes
\begin{equation}\label{eq:derivative_numerator21}
\sum_{a'\in\{0,1\}^N} \big(\hat{\Phi}_{\mathbf{A}_K}(a^\star) -  \hat{\Phi}_{\mathbf{A}_K}(a')\big)e^{\beta\hat{\Phi}_{\mathbf{A}_K}(a')} > 0.
\end{equation}
Since $a^\star$ is a maximizer of $\hat{\Phi}_{\mathbf{A}_K}$, we have that
\begin{equation}
\hat{\Phi}_{\mathbf{A}_K}(a^\star) \geq \hat{\Phi}_{\mathbf{A}_K}(a') , \ \ a' \in\{0,1\}^N.
\end{equation}
Moreover, since there is at least one $\tilde{a} \in \{0,1 \}^N$ such that $\Phi_{\mathbf{A}_K}(a^\star) > \Phi_{\mathbf{A}_K}(\tilde{a})$, we have that  \eqref{eq:derivative_numerator21} holds and consequently, $\frac{\partial g}{\partial \beta}> 0$. Therefore, the function $g(\beta,K)$ is continuous and strictly increasing in $\beta$, with $g(\beta,K) \rightarrow 1$, as $\beta\rightarrow \infty$.

To obtain the monotonicity property with respect to $K$, let $\mathbf{A}_K$ be the adjacency matrix of a fixed connected $K$-regular graph. 

Suppose $N$ is even. Then a $(K+1)$-regular graph $\mathcal{G}_{K+1}$ exists, and we denote its adjacency matrix by $\mathbf{A}_{K+1}$.
By construction, from \cref{G_K_new}, there exist permutation matrices $\Pi_1$ and $\Pi_2$ such that
\begin{equation}
\mathbf{A}_{K+1} = \Pi_2(\mathbf{A}_K + \Pi_1)\Pi_2^{\mathsf T}.
\end{equation}
Define $\tilde{a} = \Pi_2^{\mathsf T} a$, then
\begin{equation}\label{Pinvariance1}
\begin{aligned}
a^{\mathsf T}\mathbf{A}_{K+1}a
&= a^{\mathsf T}\Pi_2(\mathbf{A}_K + \Pi_1)\Pi_2^{\mathsf T} a \\[3pt]
&= \tilde{a}^{\mathsf T}(\mathbf{A}_K + \Pi_1)\tilde{a} \\[3pt]
&= \tilde{a}^{\mathsf T}\mathbf{A}_K\tilde{a} + \tilde{a}^{\mathsf T}\Pi_1\tilde{a}.
\end{aligned}
\end{equation}

Note that the $\ell^p$-induced operator norm $\|\Pi_1\|_p = 1$ for all $p$. Let $m\Equaldef \|a \|_1$. Using H\"{o}lder's inequality, we have
\begin{equation}\label{Pbound1}
\begin{aligned}
\tilde{a}^{\mathsf T} \Pi_1 \tilde{a} &\leq \|\tilde{a}\|_{\infty} \, \|\Pi_1 \tilde{a}\|_1 \\
&\leq \|\tilde{a}\|_{\infty} \, \|\Pi_1\|_1 \, \|\tilde{a}\|_1 = m.
\end{aligned}
\end{equation}
Combining \eqref{Pinvariance1} and \eqref{Pbound1} gives
\begin{equation}
    a{^\mathsf{T}}\mathbf{A}_{K+1}a \leq  \tilde{a}{^\mathsf{T}}\mathbf{A}_{K}\tilde{a} + m.
\end{equation}

Without loss of generality, assume $\theta < 1/2$. Then, the unique global maximizer of $\Phi_{\mathbf{A}_K}$ is $a^\star = \mathbb{1}$ and $\hat{\Phi}_{\mathbf{A}_K}(a^\star) = (1/2 - \theta)NK$.

For any $a \neq \mathbb{1}$, we have
\begin{equation}
    \hat{\Phi}_{\mathbf{A}_K}(a) = \frac{1}{2}a^{\mathsf{T}}\mathbf{A}_K a -K\theta m
\end{equation}
and
\begin{equation}
    \hat{\Phi}_{\mathbf{A}_{K+1}}(a) \leq \frac{1}{2}\tilde{a}^{\mathsf{T}}\mathbf{A}_K \tilde{a}+\frac{m}{2} -(K+1)\theta m.
\end{equation}

Therefore,
\begin{equation}
\hat{\Phi}_{\mathbf{A}_K}\Big(\Pi_2^\top a\Big)-\hat{\Phi}_{\mathbf{A}_{K+1}}(a) \geq \bigg( \theta-\frac{1}{2} \bigg)\|a \|_1, \ \  a \neq \mathbb{1}.
\end{equation}

Now consider
\begin{equation}
\begin{aligned}
\mu_{\mathbf{A}_{K+1}}(a^\star \mid \beta) &=
\frac{e^{\beta\hat\Phi_{\mathbf{A}_{K+1}}(a^\star)}}{\sum_{a'\in \{ 0,1\}^N}e^{\beta \hat\Phi_{\mathbf{A}_{K+1}}(a')}} \\
&=e^{\beta\big(\frac{1}{2}-\theta\big)N}\frac{e^{\beta\hat\Phi_{\mathbf{A}_K}(a^\star)}}{\sum_{a'\in \{ 0,1\}^N}e^{\beta \hat\Phi_{\mathbf{A}_{K+1}}(a')}},
\end{aligned}
\end{equation}
where we used $\hat{\Phi}_{\mathbf{A}_{K+1}}(a^\star) = \hat{\Phi}_{\mathbf{A}_K}(a^\star) + ({1}/{2} - \theta)N$.

Since $\tilde{a}' = \Pi_{2}^{\mathsf{T}}a'$ and $\Pi_2$ is bijective on $\{0,1\}^N \rightarrow \{0,1\}^N$, we have
\begin{equation}
\begin{aligned}
    \sum_{a'\in \{ 0,1\}^N}e^{\beta \hat{\Phi}_{\mathbf{A}_K}(a')}
    & = \sum_{\tilde{a}' \in \{ 0,1\}^N}e^{\beta \hat{\Phi}_{\mathbf{A}_K}(\tilde{a}')}\\
    & \geq \sum_{a'\in \{ 0,1\}^N}e^{\beta \big(\hat{\Phi}_{\mathbf{A}_{K+1}}(a') +\big( \theta-\frac{1}{2} \big)\|a' \|_1\big)}.
\end{aligned}
\end{equation}
For $a' = a^\star = \mathbb{1}$, we have $\|a'\|_1 = N$. For $a'\neq \mathbb{1}$, the bound $\hat{\Phi}_{\mathbf{A}_K}(\tilde{a}') \geq \hat{\Phi}_{\mathbf{A}_{K+1}}(a') + (\theta - 1/2)\|a'\|_1$ holds, and in particular for $a' = a^\star = \mathbb{1}$ we obtain equality.

The above inequality yields
\begin{equation}
\sum_{a'\in \{ 0,1\}^N}\!\!\!e^{\beta \hat{\Phi}_{\mathbf{A}_K}(a')} \geq e^{\beta(\theta-1/2)N}\!\!\!\sum_{a'\in \{ 0,1\}^N}\!\!\!e^{\beta \hat{\Phi}_{\mathbf{A}_{K+1}}(a')}.
\end{equation}
Rearranging, we obtain
\begin{equation}\label{eq:partition}
\sum_{a'\in \{ 0,1\}^N}\!\!\!e^{\beta \hat{\Phi}_{\mathbf{A}_{K+1}}(a')} \leq e^{\beta(1/2-\theta)N}\!\!\!\sum_{a'\in \{ 0,1\}^N}\!\!\!e^{\beta \hat{\Phi}_{\mathbf{A}_K}(a')}.
\end{equation}
Substituting back, we get
\begin{equation}\label{eq:monotonicity}
\mu_{\mathbf{A}_{K+1}}(a^\star \mid \beta) \geq \frac{e^{\beta\hat\Phi_{\mathbf{A}_K}(a^\star)}}{\sum_{a'\in \{ 0,1\}^N}e^{\beta \hat{\Phi}_{\mathbf{A}_K}(a')}} = \mu_{\mathbf{A}_K}(a^\star \mid \beta).
\end{equation}

\begin{figure}[t!]
    \centering
    \includegraphics[width=0.9\linewidth]{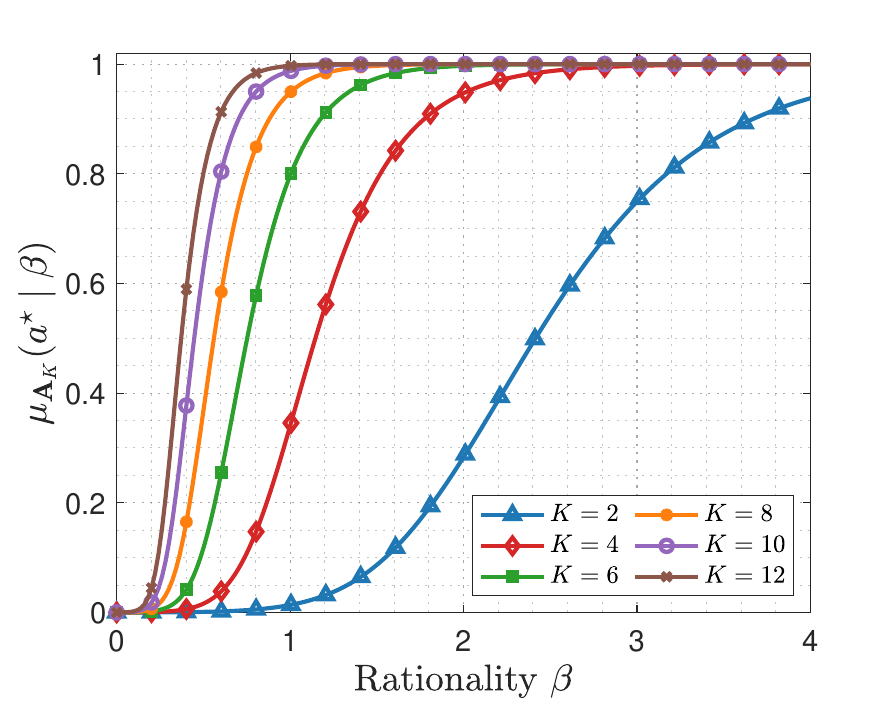}
    \includegraphics[width=0.9\linewidth]{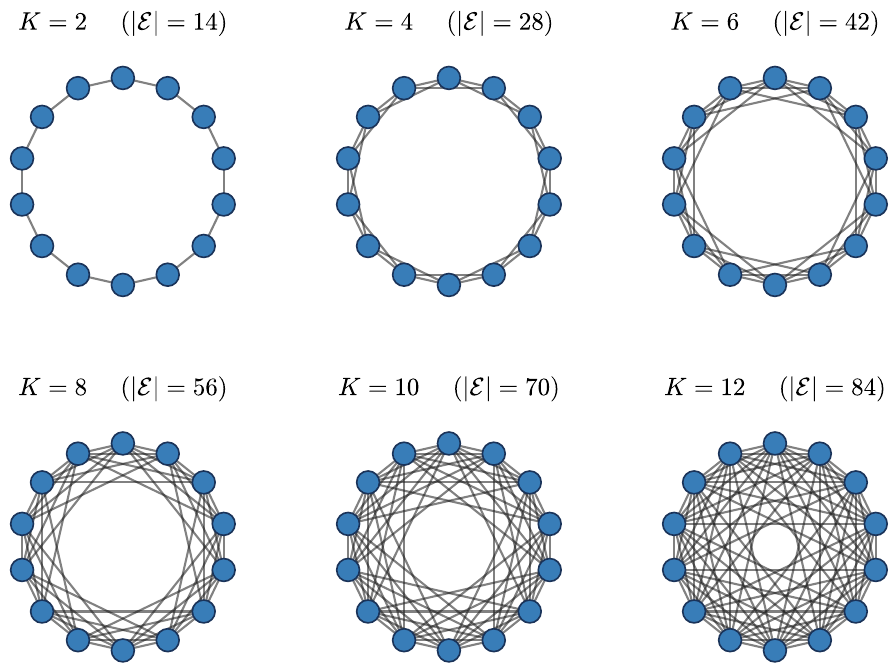}
    \caption{Stationary probability $\mu_{\mathbf{A}}(a^\star \mid \beta)$ for $K$-regular graphs with $N=14$ agents and $\theta=0.3$, as a function of $\beta$ for varying $K$. For any finite $\beta$, the probability of the risk-dominant action profile $a^\star = \mathbb{1}$ is strictly increasing in $K$, while differences vanish as $\beta \to \infty$.}
    \label{fig:monotonicity}
\end{figure}

Suppose $N$ is odd. Then $\mathcal{G}_{K+1}$ does not exist. By construction, from \cref{G_K_new}, there exists an adjacency matrix for a regular graph $\mathcal{G}_{K+2}$ given by $\mathbf{A}_{K+2} = \Pi_3(\mathbf{A}_{K}+\Pi_1+\Pi_2)\Pi_3^\top$. Using a similar inductive procedure as when $N$ is even and defining $\tilde{a} = \Pi_{3}^\top a$, then
\begin{equation}
    a{^\mathsf{T}}\mathbf{A}_{K+2}a \leq  \tilde{a}{^\mathsf{T}}\mathbf{A}_{K}\tilde{a} + 2m.
\end{equation}
The remainder of the argument proceeds identically, yielding $\mu_{\mathbf{A}_K}(a^\star \mid \beta) < \mu_{\mathbf{A}_{K+2}}(a^\star \mid \beta)$.
Therefore, the function $g$ is strictly monotone increasing in $K$.
\end{IEEEproof}
\vspace{5pt}

\begin{remark}\label{rem:strict_inequality}
The inequality \eqref{eq:monotonicity} above is in fact strict. To see this, note that for any $a'$ with $\|a'\|_1 < N$, the bound $(\theta - 1/2)\|a'\|_1 < (\theta - 1/2)N$ when $\theta < 1/2$, which introduces a strict gap in \eqref{eq:partition}. The cases $\theta > 1/2$ (where $a^\star = \mathbb{0}$) and $\theta = 1/2$ (where $a^\star = \mathbb{0} \ \text{or} \ \mathbb{1}$) follow by a analogous arguments.
\end{remark}

\vspace{5pt}

To illustrate the monotonicity results in \cref{thm:interplay} for networked coordination games with $\theta=0.3$, we evaluated $\mu_{\mathbf{A}}(a \mid \beta)$ for $K$-regular graphs with $N=14$ agents as a function $\beta$ and different values of $K$. Figure \ref{fig:monotonicity} shows how for any fixed $\beta<\infty$, the stationary probability of $a^\star =\mathbb{1}$ is strictly increasing in $K$. Also notice that in the limit of $\beta\rightarrow \infty$ of Figure \ref{fig:monotonicity}, the network connectivity does not make a significant difference as far as the stationary probability distribution.

\subsection{Minimum rationality required for coordination}

Recall the definition of $\beta^{\min}_{\mathbf{A}}(\delta)$ in \eqref{Betamin1}. Suppose that $\delta$, the probability that agents fail to coordinate on the risk-dominant equilibrium, is fixed. Then there exists a trade-off between the minimal rationality $\beta^{\min}_{\mathbf{A}}(\delta)$ and the connectivity $K$ since $\mu_{\mathbf{A}_K}(a^\star \mid \beta)$ is increasing in both $\beta$ and $K$. Intuitively, for $\theta \neq {1}/{2}$, a more connected network allows for a smaller $\beta^{\min}_{\mathbf{A}}(\delta)$ in order to guarantee that LLL achieves the same probability of  coordination  on $a^\star$.

\vspace{5pt}

\begin{theorem}\label{thm:beta_bound}
Suppose LLL is performed on a networked coordination game over a connected $K$-regular 
graph with task difficulty $\theta \neq {1}/{2}$. Then,
\begin{multline}\label{Boundforbeta}
\beta_{\mathbf{A}_K}^{\min}(\delta) \leq \left|\left(\frac{1}{2}-\theta\right)K\right|^{-1} 
\times \\ \left(\frac{\log(1-\delta)}{N}-\log\left(1-e^{\frac{\log(1-\delta)}{N}}\right)\right).
\end{multline}
\end{theorem}

\vspace{5pt}

\begin{IEEEproof}
First, notice that
\begin{equation}\label{mu_K}
\mu_{\mathbf{A}_K}(a^\star \mid \beta) = \frac{e^{\beta\big( \frac{1}{2}a^{\star\mathsf{T}}\mathbf{A}_K a^\star -K\theta\mathbb{1}^{\mathsf{T}}a^\star \big)}}{\sum_{a'\in \{ 0,1\}^N}e^{\beta \big( \frac{1}{2}a'^{\mathsf{T}}\mathbf{A}_K a'-K\theta\mathbb{1}^\mathsf{T}a' \big) }}.
\end{equation}
From \cref{lem:bound_quadratic_form} and the fact that $\beta \geq 0$, we have
\begin{equation}
    e^{\beta \big(-K\theta\mathbb{1}^\mathsf{T}a' + \frac{1}{2}a'^{\mathsf{T}}\mathbf{A}_K a' \big) } \leq e^{\beta \big(-K\theta m + \frac{1}{2}{mK} \big) },
\end{equation}
where $\|a'\|_1 = m$. Since $\exp(\cdot)$ is a strictly increasing function, for all $a' \in \{0,1\}^N$, we can group terms by their Hamming weight to obtain
\begin{equation}\label{Binomial_bound}
    \sum_{a'\in \{ 0,1\}^N}\!\!\!\!\!\!\!\!e^{\beta(-K\theta\mathbb{1}_N^\mathsf{T}a' + \frac{1}{2}a'^{\mathsf{T}}\mathbf{A}_K a') } \leq \sum_{m=0}^{N} \binom{N}{m}e^{\beta (-K\theta m + \frac{1}{2}mK) }.
\end{equation}
Applying the Binomial Theorem to the right-hand side of \eqref{Binomial_bound}, we obtain
\begin{equation}
\sum_{m=0}^{N} \binom{N}{m}e^{\beta K(\frac{1}{2}-\theta) m} = \Big(1+e^{\beta K(\frac{1}{2}-\theta)}\Big)^N.
\end{equation}
Therefore, a lower bound on \eqref{mu_K} is given by
\begin{equation}\label{LB}
\begin{aligned}
     \mu_{\mathbf{A}_K}(a^\star \mid \beta)
     &\geq \frac{e^{\beta K(\frac{1}{2}-\theta)N}}{\big(1+e^{\beta K(\frac{1}{2}-\theta)}\big)^N} \\
     &= \left(\frac{1}{1+e^{-\beta K(\frac{1}{2}-\theta)}}\right)^{\!N}.
\end{aligned}
\end{equation}
    The proof follows immediately from setting the right-hand side of \eqref{LB} equal to $1-\delta$ and solving for $\beta$.
\end{IEEEproof}

\vspace{5pt}

\begin{remark}
     Note that the right-hand side of \eqref{LB} is also an increasing function of $\beta$, which can be verified by taking its derivative with respect to $\beta$. Moreover, this lower bound matches the true value of $\mu_{\mathbf{A}_K}(a^\star \mid \beta)$ when $\beta = 0$ and $\beta \rightarrow \infty$, so the bound in \eqref{LB} is asymptotically tight.  
\end{remark}

\vspace{5pt}

\begin{corollary}\label{cor:beta_min}
Suppose LLL is performed on a networked coordination game over a connected $K$-regular graph with task difficulty $\theta \neq {1}/{2}$. Then,
\begin{equation}
\beta_{\mathbf{A}_K}^{\min}(\delta) \propto \frac{1}{K}.
\end{equation}
\end{corollary}

\vspace{5pt}

\begin{IEEEproof}
The proof follows from \Cref{thm:beta_bound} and the definition in \eqref{Betamin1}.
\end{IEEEproof}

\vspace{5pt}

The consequence of \Cref{thm:beta_bound} and \Cref{cor:beta_min}  is that when agents are involved in the networked coordination game, {\it more connected agents can afford to be less rational then less connected ones}. The upper bound and the true value (obtained numerically) of $\beta^{\mathrm{min}}_{\mathbf{A}_K}(\delta)$ are shown in \Cref{fig:betamin} for different values of $K$.
\begin{figure}
    \centering
\includegraphics[width=1\linewidth]{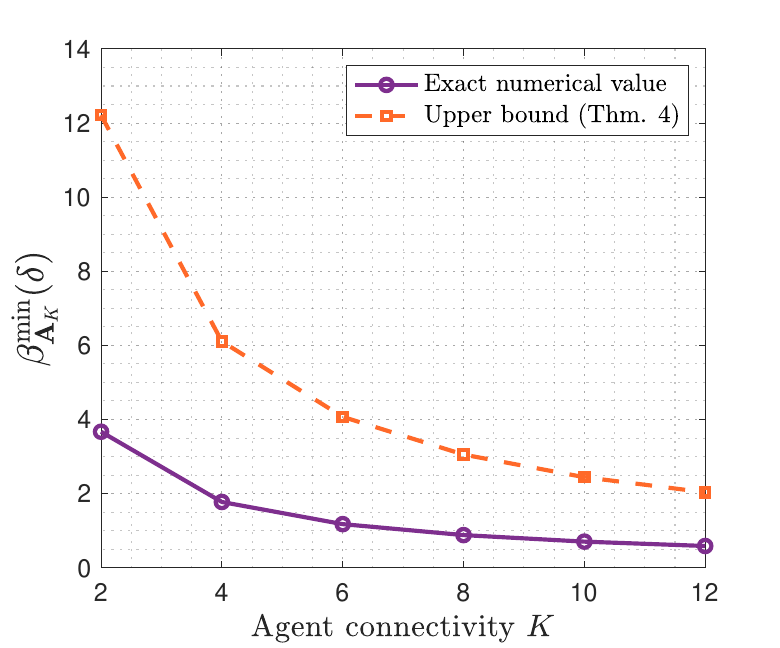}
    \caption{Upper bound and numerical value of $\beta^{\min}_{\mathbf{A}_K}(\delta)$ versus $K$ (\Cref{thm:beta_bound,cor:beta_min}). Higher connectivity reduces the rationality required for coordination.}
    \label{fig:betamin}
\end{figure}

\vspace{5pt}

\section{Irregular Graphs}\label{sec:irregular}
Having characterized the effect of connectivity on the stationary 
probability of jointly selecting the risk-dominant equilibrium $a^\star$ for regular graphs, we extend the analysis to irregular graphs. We show that the coordination probability remains monotone in the number of edges in this more general setting.

\subsection{Inductive improvement by edge augmentation}
An important measure of graph connectivity is the number of edges. As 
shown in the next theorem, the stationary probability of LLL learning to play the optimal action profile grows monotonically with the number of edges.

\vspace{5pt}

\begin{definition}\label{def:successor}
Graph $\mathcal{G}_s = ([N], \mathcal{E} \cup \{(i,j)\})$ is called a \emph{successor} 
of graph $\mathcal{G} = ([N], \mathcal{E})$ if $(i,j) \notin \mathcal{E}$.
\end{definition}

\vspace{5pt}

\begin{theorem}\label{inductiveimprove}
Let $\mathcal{G}_s$ be a successor of $\mathcal{G}$. Then, for any $\beta > 0$,
\begin{equation}
    \mu_{\mathbf{A}_s}(a^\star \mid \beta) > \mu_{\mathbf{A}}(a^\star \mid \beta),
\end{equation}
where $\mathbf{A}$ and $\mathbf{A}_s$ are the adjacency matrices of $\mathcal{G}$ and 
$\mathcal{G}_s$, respectively.
\end{theorem}
\vspace{5pt}
\begin{IEEEproof}
    Let $\mathbf{A}$ and $\mathbf{A}_s$ denote the adjacency matrices of $\mathcal{G}$ and $\mathcal{G}_s$, respectively. We have
    \begin{equation}
        \mathbf{A}_s - \mathbf{A} = \mathbb{e}_i\mathbb{e}_j^\mathsf{T} + \mathbb{e}_j\mathbb{e}_i^\mathsf{T},
    \end{equation}
    where $\mathbb{e}_i$ and $\mathbb{e}_j$ are the $i$-th and the $j$-th standard basis vectors in $\mathbb{R}^N$.

    We prove the theorem for the case $\theta < {1}/{2}$, so that $a^\star = \mathbb{1}$. The case $\theta > {1}/{2}$ is symmetric, and $\theta = {1}/{2}$ is simpler, thus are omitted here.
    For $a^\star = \mathbb{1}$, the potential values on $\mathcal{G}_s$ and $\mathcal{G}$ satisfy
    \begin{equation}
        \Phi_s(a^\star) - \Phi(a^\star) = \Big(\frac{1}{2} - \theta\Big) \mathbb{1}^\mathsf{T}(e_ie_j^\mathsf{T}+e_je_i^\mathsf{T})\mathbb{1} = 2\Big(\frac{1}{2} - \theta\Big).
    \end{equation}
    However, for all $a \in \mathcal{A}$ such that $a_i = 0$ or $a_j = 0$, we have
    \begin{equation}
        \Phi_s(a) = \Phi(a).
    \end{equation}
    The set $\{a\in \mathcal{A} \mid a_i = 0 \text{ or } a_j = 0\}$ has cardinality $3\cdot 2^{N-2}$ and is never empty. Then, we can compare $\mu_{\mathcal{G}_s}(a^\star \mid \beta)$ and $\mu_{\mathcal{G}}(a^\star \mid \beta)$ as follows
    \begin{equation} \label{eq:denominator}
        \begin{aligned}
\mu_\mathcal{G}(a^\star \mid \beta)
&=
\frac{e^{\beta\Phi(a^\star)}e^{\beta(1-2\theta)}}{\sum_{a\in \mathcal{A}}e^{\beta \Phi(a)}e^{\beta(1-2\theta)}} \\
&<
\frac{e^{\beta\Phi_s(a^\star)}}{\sum_{a\in \mathcal{A}}e^{\beta \Phi_s(a)}} = \mu_{\mathcal{G}_s}(a^\star \mid \beta).
\end{aligned}
    \end{equation}
    The inequality is strict since there exists at least one $a$ with $a_i = 0$ or $a_j = 0$ such that
    \begin{equation}
        e^{\beta\Phi(a)}e^{\beta(1-2\theta)} > e^{\beta\Phi_s(a)},
    \end{equation}
    which increases the denominator on the left-hand side relative to the right-hand side of \eqref{eq:denominator}, while both expressions share the same numerator $e^{\beta\Phi_s(a^\star)}$.
\end{IEEEproof}
\vspace{5pt}

\cref{inductiveimprove} states that an increase in edge connectivity reduces the value of $\beta^{\min}_{\mathbf{A}}(\delta)$. This can be visualized for a system with $N=14$ agents in \cref{fig:monotonicity_irregular} where edges are randomly placed between two previously disconnected agents. That observation leads naturally to the question of how to distribute edges among a set of agents such that we maximize the stationary probability of coordination.

\begin{figure}
    \centering
    \includegraphics[width=1\linewidth]{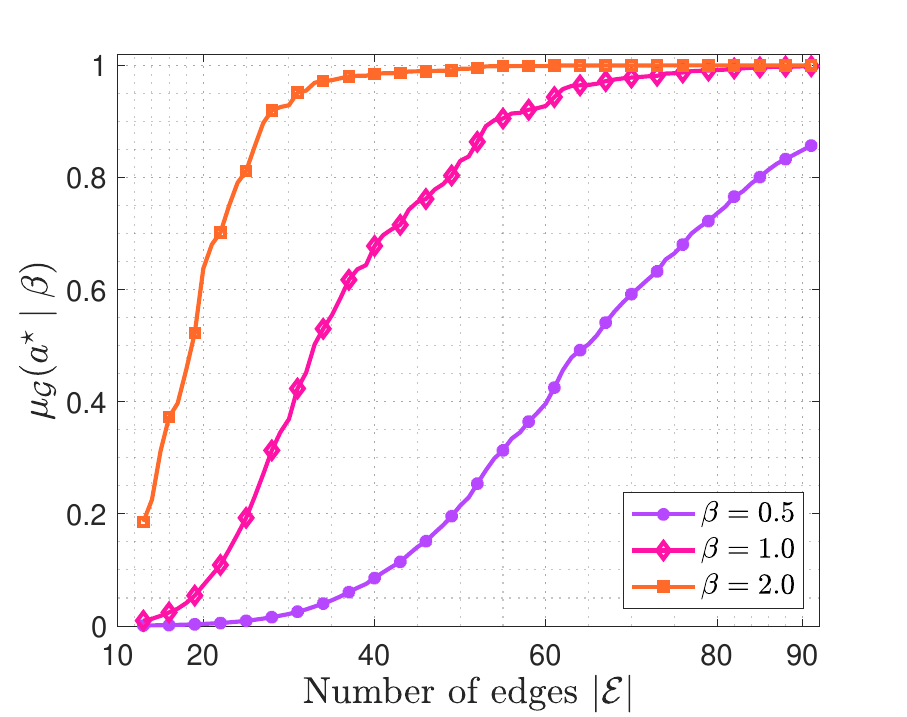}
    \caption{Coordination probability $\mu_{\mathcal{G}}(a^\star \mid \beta)$ versus 
number of edges $|\mathcal{E}|$ for $N=14$ agents and a coordination game with $\theta=0.3$. 
Adding edges monotonically increases the probability of coordination.}
    \label{fig:monotonicity_irregular}
\end{figure}

\section{Optimal Network Design}\label{sec:gaussian}

Consider the problem of constructing a network with a fixed number of 
edges $|\mathcal{E}|$ among $N$ boundedly rational agents using LLL with a fixed parameter $\beta$, where the objective is to maximize the 
stationary probability of jointly selecting the risk-dominant action 
profile $a^\star$. This problem is in general NP-hard due to the 
combinatorial explosion in the number of possible graphs and the non-convexity of the objective function. Moreover, evaluating the objective requires computing a sum of cardinality $2^N$, which is impractical even for networks of moderate size. Nevertheless, in this section we characterize the role of regular graphs in three regimes: (1) small rationality; (2) moderate rationality; and (3) asymptotically large networks, $N \to \infty$.

\subsection{Ising stag hunt game reparameterization}

We reparametrize our game using Rademacher variables $s_i = 2a_i - 1 \in \{-1,1\}$, obtaining an Ising game \cite{leonidov2024ising} with the following equivalent potential function,
\begin{equation}
\tilde{\Phi}_{\mathbf{A}}(s)
= \frac{1}{8}\, s^\mathsf{T} \mathbf{A} s
+ \left(\frac{1}{4} - \frac{\theta}{2}\right)\, \mathbb{1}^\mathsf{T} \mathbf{A} s
+ \frac{1}{8}\, \mathbb{1}^\mathsf{T} \mathbf{A} \mathbb{1},
\end{equation}
where $s \in \{-1,1\}^N.$ This reparameterization symmetrizes the state space and simplifies the  analysis.

Disregarding the constant term, we write
\begin{equation}
\tilde{\Phi}_{\mathbf{A}}(s)
= \frac{1}{8}\, s^\mathsf{T} \mathbf{A} s
+ \left(\frac{1}{4} - \frac{\theta}{2}\right)\, \mathbb{1}^\mathsf{T} \mathbf{A} s,
\end{equation}
which leads to the following stationary probability for an optimal strategy profile $s^\star \in \{ -\mathbb{1},\mathbb{1}\}$,
\begin{equation}\label{eq:optimization}
\tilde{\mu}_{\mathbf{A}}(s^\star \mid \beta) \Equaldef
\frac{e^{\beta \tilde{\Phi}_{\mathbf{A}}(s^\star)}}
{\sum_{s' \in \{-1,1\}^N} e^{\beta \tilde{\Phi}_{\mathbf{A}}(s')}} .
\end{equation}

We are interested in solving the following optimization 
problem
\begin{equation}\label{eq:opt_problem}
    \mathbf{A}^\star \in \operatorname*{arg\,max}_{\mathbf{A} \,\in\, 
    \mathcal{G}(N,|\mathcal{E}|)} 
    \tilde{\mu}_{\mathbf{A}}(s^\star \mid \beta),
\end{equation}
where $\mathcal{G}(N,|\mathcal{E}|)$ denotes the set of all connected 
simple graphs on $N$ nodes with  $|\mathcal{E}|$ edges.

\subsection{Partition function as a moment generating function}

In statistical physics, the denominator of the Gibbs distribution is known as the \textit{partition function} \cite{mezard2009information}. The key tool for optimizing over graph structures is the observation that the partition function can be expressed in terms of a moment generating function (MGF).  The partition function in the Rademacher coordinates is
\begin{equation}\label{eq:Z_mgf}
Z_\mathbf{A}(\beta)
\Equaldef \sum_{s\in \{-1,1\}^N} e^{\beta\tilde{\Phi}_{\mathbf{A}}(s)}
= 2^N\, \mathbb{E}_S\!\left[e^{\beta\tilde{\Phi}_{\mathbf{A}}(S)}\right],
\end{equation}
where $S$ is a uniformly distributed random vector taking values on the set $\{-1,1\}^N$, i.e.,
\begin{equation}
\mathbb{P}(S=s) = \frac{1}{2^N}, \ \ s\in \{-1,1\}^N.
\end{equation}
The expectation on the right-hand side is the MGF of $\tilde{\Phi}_{\mathbf{A}}(S)$ evaluated at $\beta$.

We first observe that the numerator of the stationary distribution at the optimal action profile depends only on the number of edges and it is independent of graph's degree distribution.

\vspace{5pt}

\begin{lemma}\label{lem:numerator_invariance}
Let $\mathcal{G}$ be a simple undirected graph with $|\mathcal{E}|$ edges.
In the Rademacher parametrization, the potential at $s^\star = \mathbb{1}$ and at $s^\star = -\mathbb{1}$  are
\begin{equation}
\tilde{\Phi}_{\mathbf{A}}(\mathbb{1}) = \Big(\frac{3}{4} - \theta\Big)|\mathcal{E}|,
\qquad
\tilde{\Phi}_{\mathbf{A}}(-\mathbb{1}) = \Big(\theta-\frac{1}{4} \Big)|\mathcal{E}|,
\end{equation}
Therefore, $\tilde{\Phi}_{\mathbf{A}}(s^\star)$ depends on $\mathbf{A}$ only through $|\mathcal{E}|$.
\end{lemma}

\vspace{5pt}

\begin{IEEEproof}
The proof follows from direct computation.
\end{IEEEproof}

\vspace{5pt}

Since $e^{\beta\tilde{\Phi}_{\mathbf{A}}(s^\star)}$ is constant for all graphs with the same number of edges $|\mathcal{E}|$, maximizing the stationary probability of $s^\star$ given by
\begin{equation}\label{eq:mu_over_Z}
\tilde{\mu}_{\mathbf{A}}(s^\star \mid \beta) = \frac{e^{\beta\tilde{\Phi}_{\mathbf{A}}(s^\star)}}{Z_\mathbf{A}(\beta)}
\end{equation}
is equivalent to minimizing $Z_\mathbf{A}(\beta)$, or equivalently, minimizing $\mathbb{E}_S[e^{\beta\tilde{\Phi}_{\mathbf{A}}(S)}]$.
This equivalence holds for all values of $\beta$.

\subsection{Optimality of regular graphs for small rationality}

We first consider case when the agent's rationality $\beta$ is small.  Expanding the MGF in a Taylor series around $\beta = 0$, we obtain
\begin{equation}
\mathbb{E}_S\!\Big[e^{\beta\tilde{\Phi}_{\mathbf{A}}(S)}\Big]
= 1 + \beta\,\mathbb{E}\big[\tilde{\Phi}_{\mathbf{A}}(S)\big] + \frac{\beta^2}{2}\,\mathbb{E}\big[\tilde{\Phi}_{\mathbf{A}}(S)^2\big] + O(\beta^3).
\end{equation}

Since $\mathbf{A}$ is the adjacency matrix of an undirected simple graph, we have
$\mathbf{A} = \mathbf{A}^\mathsf{T}$ and $A_{ii} = 0$.
For $S \in \{-1,1\}^N$ uniformly distributed, we have $\mathbb{E}[S_i] = 0$ and
$\mathbb{E}[S_i S_j] = 0$ for all $i \neq j$. Therefore,
\begin{equation}
\mathbb{E}[S^\mathsf{T}\mathbf{A}S]
= \sum_{i \neq j} A_{ij}\, \mathbb{E}[S_i S_j]
= 0
\end{equation}
and
\begin{equation}
\mathbb{E}[\mathbb{1}^\mathsf{T}\mathbf{A}S]
= \sum_{i=1}^N d_i\, \mathbb{E}[S_i]
= 0.
\end{equation}
Computing the first and second moments of $\tilde{\Phi}_{\mathbf{A}}(S)$, we get
\begin{equation}
\mathbb{E}\big[\tilde{\Phi}_{\mathbf{A}}(S)\big]
= \frac{1}{8}\,\mathbb{E}[S^\mathsf{T}\mathbf{A}S]
+ \left(\frac{1}{4} - \frac{\theta}{2}\right)
  \mathbb{E}[\mathbb{1}^\mathsf{T}\mathbf{A}S]
= 0
\end{equation}
and, after some algebra and using properties of Rademacher random variables, we obtain 
\begin{equation}
\mathbb{E}\big[\tilde{\Phi}_{\mathbf{A}}(S)^2\big] = \frac{|\mathcal{E}|}{16} 
+ \left(\frac{1}{4}-\frac{\theta}{2}\right)^2\sum_{i=1}^N d_i^2 \Equaldef \sigma_{\mathbf{A}}^2.
\end{equation}
Finally,
\begin{equation}\label{eq:Z_small_beta}
\mathbb{E}_S\!\Big[e^{\beta\tilde{\Phi}_{\mathbf{A}}(S)}\Big]
= 1 + \frac{\beta^2}{2}\sigma_{\mathbf{A}}^2+ O(\beta^3).
\end{equation}

\vspace{5pt}

\begin{theorem}\label{thm:small_beta}
For sufficiently small $\beta > 0$, the stationary probability 
$\tilde{\mu}_{\mathbf{A}}(a^\star \mid \beta)$ is maximized, over all graphs 
on $N$ vertices with $|\mathcal{E}|$ edges, by the $K$-regular graph 
with $K = 2|\mathcal{E}|/N$, or by a near-$K$-regular graph when 
$2|\mathcal{E}|/N$ is not an integer.
\end{theorem}

\vspace{5pt}

\begin{IEEEproof}
From \eqref{eq:Z_small_beta}, the partition function is
\begin{equation}
Z_\mathbf{A}(\beta) \approx 2^N\Big(1 + \frac{\beta^2}{2}\sigma_\mathbf{A}^2\Big),
\end{equation}
which is monotone increasing in $\sigma_\mathbf{A}^2$. Therefore, we are interested in minimizing the variance 
\begin{equation}
\sigma_\mathbf{A}^2 = \frac{|\mathcal{E}|}{16} + \left(\frac{1}{4}-\frac{\theta}{2}\right)^2 \sum_{i=1}^N d_i^2.
\end{equation}
The first term is fixed for a given $|\mathcal{E}|$. For the second 
term, minimizing $\sigma_{\mathbf{A}}^2$ is equivalent to minimizing
$\sum_{i=1}^N d_i^2$
over all degree sequences $(d_1,\ldots,d_N)$ with 
$\sum_{i=1}^N d_i = 2|\mathcal{E}|$. We use \textit{Majorization theory} 
\cite{marshall2011inequalities} to characterize the minimizer. Recall 
that a vector $\mathbf{x}$ is \emph{majorized} by $\mathbf{y}$, i.e., 
$\mathbf{x} \preceq \mathbf{y}$, if for all $k = 1,\ldots,N$, we have
\begin{equation}\label{eq:majorization1}
    \sum_{i=1}^k x_{[i]} \leq \sum_{i=1}^k y_{[i]}, 
    \qquad \sum_{i=1}^N x_i = \sum_{i=1}^N y_i,
\end{equation}
where $\xi_{[1]} \geq \xi_{[2]} \geq \cdots \geq \xi_{[N]}$ denotes the decreasing 
rearrangement of a vector $\boldsymbol{\xi}\in \mathbb{R}^N$. Since $f(d) = d^2$ is convex, it is also \emph{Schur-convex}, i.e.,
\begin{equation}\label{eq:majorization2}
    \mathbf{x} \preceq \mathbf{y} \implies 
    \sum_{i=1}^N f(x_i) \leq \sum_{i=1}^N f(y_i).
\end{equation}
Therefore, $\sum_{i=1}^N d_i^2$ is minimized by the least majorized degree 
sequence, i.e., the most uniform one. Among all non-negative integer 
sequences with fixed sum $2|\mathcal{E}|$:
\begin{itemize}
    \item When $K = 2|\mathcal{E}|/N$ is an integer, the least majorized 
    sequence is the constant sequence $(K,\ldots,K)$, which corresponds to 
    a $K$-regular graph.
    \item When $K=2|\mathcal{E}|/N$ is not an integer, the minimizer is a 
    near-$K$-regular sequence in which each 
    $d_i \in \{\lfloor K\rfloor, \lceil K\rceil\}$, since any other 
    sequence with the same sum is majorized by it.
\end{itemize}
Therefore, $\sigma_{\mathbf{A}}^2$ is minimized by the $K$-regular (or 
near-$K$-regular) graph. From \cref{lem:numerator_invariance}, $\tilde{\Phi}_{\mathbf{A}}(a^\star)$ depends only on 
$|\mathcal{E}|$ and not on the graph structure, the numerator of 
$\tilde{\mu}_{\mathbf{A}}(a^\star \mid \beta)$ is identical across all graphs 
with the same $|\mathcal{E}|$. For sufficiently small $\beta$,
\begin{equation}
    \tilde{\mu}_{\mathbf{A}}(a^\star \mid \beta) \approx \frac{e^{\beta\tilde{\Phi}_{\mathbf{A}}(s^\star)}}{
    2^N\left(1 
    + \frac{\beta^2}{2}\sigma_{\mathbf{A}}^2 \right)},
\end{equation}
which is decreasing in $\sigma_{\mathbf{A}}^2$. Consequently, minimizing 
$\sigma_{\mathbf{A}}^2$ maximizes $\tilde{\mu}_{\mathbf{A}}(a^\star \mid \beta)$ 
for small $\beta$, and the $K$-regular graph is the optimizer.
\end{IEEEproof}

\vspace{5pt}

\begin{remark}
\cref{thm:small_beta} reveals that, in the low-rationality regime, the graph structure affects the stationary distribution of coordination only through its degree distribution, while higher-order graph invariants (triangles, spectral gap, etc.) do not play a significant role and can be ignored.
\end{remark}

\subsection{Moderate rationality}\label{sec:moderate_beta}

For moderate values of $\beta$, the Taylor expansion is no longer accurate. We instead employ a spectral upper bound on the partition function that is valid for any $\beta>0$ and $N\geq 2$.

\vspace{5pt}
\begin{theorem}\label{thm:moderate_beta}
Among all simple connected graphs on $N$ vertices with $|\mathcal{E}|$ 
edges, the coordination probability under LLL satisfies
\begin{equation}\label{eq:mu_spectral_lb}
\mu_{\mathbf{A}}(a^\star \mid \beta) \geq \left(\frac{1}{1 + 
e^{-\frac{\beta \lambda_1(\mathbf{A})}{2}|1-2\theta|}}\right)^{\!N},
\end{equation}
for all $\beta > 0$. This lower bound is maximized over all graphs with 
$|\mathcal{E}|$ edges by the $K$-regular graph with $K = 2|\mathcal{E}|/N$ 
(when it exists). Therefore, the optimal graph $\mathbf{A}^\star$ satisfies
\begin{equation}\label{eq:mu_spectral_lb_regular}
\mu_{\mathbf{A}^\star}(a^\star \mid \beta) \geq \left(\frac{1}{1 + 
e^{-\frac{\beta K}{2}|1-2\theta|}}\right)^{\!N}.
\end{equation}
\end{theorem}

\vspace{5pt}

\begin{IEEEproof}
Working in the original binary action coordinates, recall the potential function
\begin{equation}
\Phi_{\mathbf{A}}(a) = \frac{1}{2}a^\mathsf{T}\mathbf{A}a - \theta\,\mathbb{1}^\mathsf{T}\mathbf{A}a.
\end{equation}
Completing the square, we obtain
\begin{equation}\label{eq:complete_square}
\Phi_{\mathbf{A}}(a) = \frac{1}{2}y^\mathsf{T}\mathbf{A}y - \theta^2|\mathcal{E}|,
\end{equation}
where $y \Equaldef a - \theta\mathbb{1}$. 

By the Rayleigh quotient characterization of the largest eigenvalue \cite{horn2012matrix}, we have
\begin{equation}\label{eq:rayleigh}
y^\mathsf{T}\mathbf{A}y \leq \lambda_1(\mathbf{A})\,\|y\|_2^2, \ \   y\in\mathbb{R}^N,
\end{equation}
where $\lambda_1(\mathbf{A})$ denotes the denotes the
largest eigenvalue of $\mathbf{A}$. For an action profile $a$ with $m\Equaldef \|a\|_1$, we have
\begin{equation}\label{eq:y_norm}
\|y\|_2^2 = m(1-\theta)^2 + (N-m)\theta^2.
\end{equation}
Substituting \eqref{eq:rayleigh} and \eqref{eq:y_norm} into \eqref{eq:complete_square}, we obtain the following upper bound for the potential function
\begin{equation}\label{eq:phi_spectral_ub}
\Phi_{\mathbf{A}}(a) \leq \frac{\lambda_1(\mathbf{A})}{2}\big(m(1-\theta)^2 + (N-m)\theta^2\big) - \theta^2|\mathcal{E}|.
\end{equation}

\begin{figure}[ht!]
    \centering
    \includegraphics[width=1\linewidth]{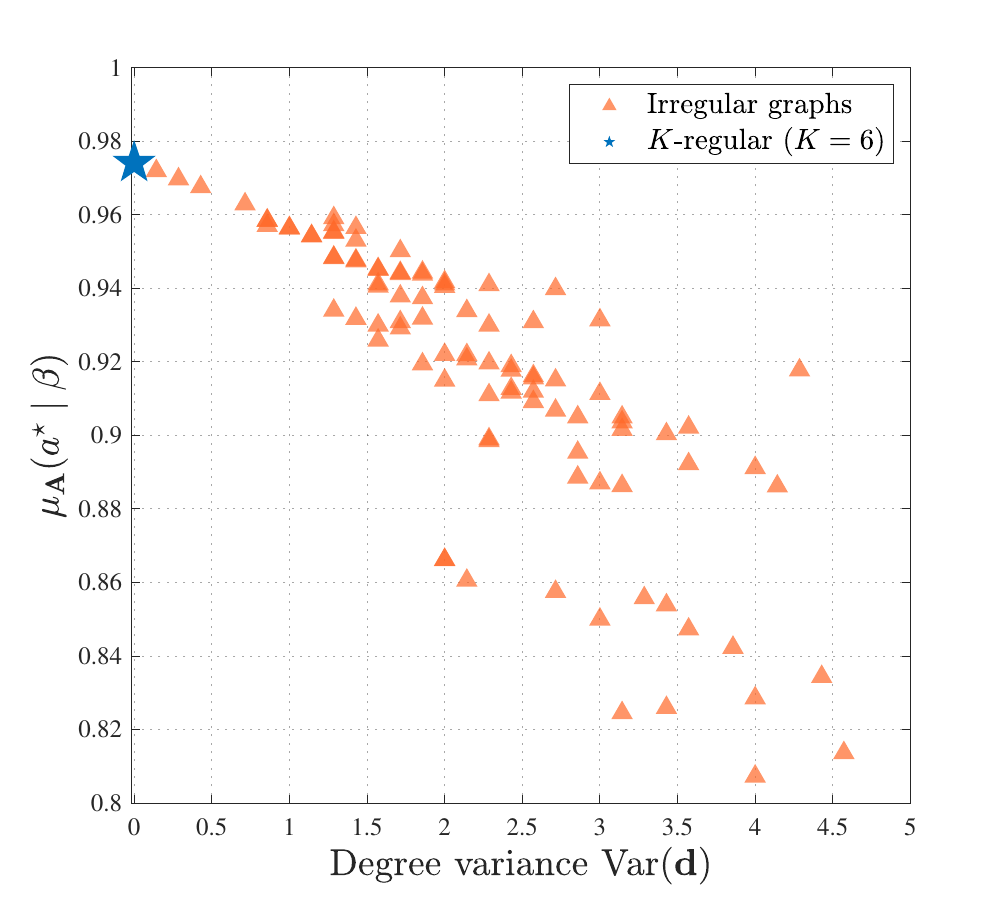}
    \caption{Coordination probability $\mu_{\mathbf{A}}(a^\star \mid \beta)$ versus 
degree variance $\mathrm{Var}(d)$ for irregular and $K$-regular graphs 
with $N=14$ nodes and $|\mathcal{E}|=42$ edges. The $K$-regular graph 
achieves the highest coordination probability.}
    \label{fig:regular_vs_irregular}
\end{figure}

The bound \eqref{eq:phi_spectral_ub} depends on $a$ only through 
$m = \|a\|_1$. Hence,
taking the exponential and summing over all $a \in \{0,1\}^N$, we get
\begin{equation}
    Z_{\mathbf{A}}(\beta) \leq 
    e^{\beta\left(\frac{\lambda_1(\mathbf{A}) N\theta^2}{2} - \theta^2|\mathcal{E}|\right)}\! \! \!
    \sum_{a \in \{0,1\}^N} 
    e^{\frac{\beta\lambda_1(\mathbf{A})}{2}\left[(1-\theta)^2-\theta^2\right]\|a\|_1}.
\end{equation}
Since
\begin{multline}
    \sum_{a \in \{0,1\}^N} 
    e^{\frac{\beta\lambda_1(\mathbf{A})}{2}\left[(1-\theta)^2-\theta^2\right]\|a\|_1}
    \\ = \sum_{m=0}^{N} \binom{N}{m} 
    e^{\frac{\beta\lambda_1(\mathbf{A})}{2}\left[(1-\theta)^2-\theta^2\right]m}.
\end{multline}
Writing 
\begin{equation}
e^{\frac{\beta\lambda_1(\mathbf{A})}{2}\left[(1-\theta)^2-\theta^2\right]m} 
= \left(e^{\frac{\beta\lambda_1(\mathbf{A})(1-\theta)^2}{2}}\right)^m 
\left(e^{-\frac{\beta\lambda_1(\mathbf{A})\theta^2}{2}}\right)^m
\end{equation}
and multiplying and dividing by $e^{\frac{\beta\lambda_1\theta^2}{2}(N-m)}$, 
the sum becomes
\begin{equation}
    \sum_{m=0}^{N}\binom{N}{m}
    \left(e^{\frac{\beta\lambda_1(\mathbf{A}) (1-\theta)^2}{2}}\right)^m
    \left(e^{\frac{\beta\lambda_1(\mathbf{A}) \theta^2}{2}}\right)^{N-m}.
\end{equation}
Using the Binomial Theorem yields
\begin{equation}\label{eq:Z_spectral}
Z_\mathbf{A}(\beta) \leq e^{\beta\big(\frac{\lambda_1(\mathbf{A}) N\theta^2}{2}-\theta^2|\mathcal{E}|\big)} \Big(e^{\frac{\beta\lambda_1(\mathbf{A})\theta^2}{2}} + e^{\frac{\beta\lambda_1(\mathbf{A})(1-\theta)^2}{2}}\Big)^N.
\end{equation}

From the Rayleigh quotient, we have that for any graph with $|\mathcal{E}|$ edges, the largest eigenvalue satisfies the following inequality
\begin{equation}\label{eq:lambda1_bound}
\lambda_1(\mathbf{A}) \geq \frac{\mathbb{1}^\mathsf{T} \mathbf{A}\, \mathbb{1}}{\|\mathbb{1}\|_2^2} = \frac{2|\mathcal{E}|}{N}
\end{equation}
with equality if and only if $\mathbf{A}$ is the adjacency matrix of a $K$-regular graph with $K = 2|\mathcal{E}|/N$.

The upper bound in \eqref{eq:Z_spectral} is an increasing function of $\lambda_1(\mathbf{A})$ for all $\beta > 0$.  By \eqref{eq:lambda1_bound}, $\lambda_1$ is minimized by the $K$-regular graph, so the upper bound is also minimized when the graph is $K$-regular.

To obtain \eqref{eq:mu_spectral_lb}, consider $a^\star = \mathbb{1}$ (the case when $a^\star = \mathbb{0}$ is analogous).  Then, $\Phi_{\mathbf{A}_K}(\mathbb{1}) = K(1/2 - \theta)N$.  Using $\hat{\Phi}_{\mathbf{A}_K}(a) = \Phi_{\mathbf{A}_K}(a) - \theta^2|\mathcal{E}|$, we have
\begin{multline}
\mu_{\mathbf{A}^\star}(\mathbb{1}\mid\beta) \geq \frac{e^{\beta K(1-\theta)^2 N/2}}{\big(e^{\frac{\beta K\theta^2}{2}} + e^{\frac{\beta K(1-\theta)^2}{2}}\big)^N} \\ = \left(\frac{1}{1+e^{-\frac{\beta K(1-2\theta)}{2}}}\right)^{\!N}.
\end{multline}
Since $\theta < 1/2$ implies $a^\star = \mathbb{1}$, we have $1-2\theta = |1-2\theta|$.  The case $a^\star = \mathbb{0}$ gives the same expression with $2\theta-1$.
\end{IEEEproof}

\vspace{5pt}

\begin{remark}
\Cref{thm:moderate_beta} shows that using $K$-regular graph maximizes a 
lower bound on $\mu_{\mathbf{A}^\star}(a^\star \mid \beta)$ that depends on 
the graph only through $\lambda_1(\mathbf{A})$. Whether 
exact optimality holds for all $\beta$ remains an open question. \Cref{fig:regular_vs_irregular} shows that when compared with randomly generated connected irregular graphs with a fixed number of edges, the $K$-regular graph yields the largest stationary probability of coordination.

The spectral bound reveals two distinct ways by which regularity improves coordination. First, the spectral radius $\lambda_1$ is minimized, giving the tightest Rayleigh quotient bound on the quadratic form $y^\mathsf{T}\mathbf{A}y$. Second, the leading term in 
\eqref{eq:Z_spectral} is equal to one: the condition 
$\lambda_1 N/2 = |\mathcal{E}|$ holds if and only if the graph is regular, and for irregular graphs the leading factor 
$e^{\beta(\lambda_1 N\theta^2/2 - \theta^2|\mathcal{E}|)} > 1$, increasing the partition function and reducing the coordination probability.

For $K$-regular graphs, the lower bound in \eqref{eq:mu_spectral_lb} can be compared with the bound in \eqref{LB}.  Both have the sigmoid-power form, but they arise from different bounding techniques: \eqref{LB} uses the operator norm bound on $a^\mathsf{T}\mathbf{A}_K a \leq mK$ applied directly to the potential, while \eqref{eq:mu_spectral_lb} uses the Rayleigh quotient after completing the square.  For $\theta \in (0,1)$, the two bounds are equivalent and yield the same minimum-$\beta$ condition of \Cref{thm:beta_bound}.
\end{remark}


\subsection{Optimization of asymptotically large graphs}

For arbitrary $\beta$, the partition function depends on more than just the graph's degree distribution, and as a result the optimization becomes extremely challenging. Here, we show that tractability is recovered in the regime of large graphs, when $N\rightarrow \infty$. Under mild technical conditions, we can show that the partition function can be well approximated by the MGF of a Gaussian random variable. This is illustrated in \cref{fig:MCLT}. We begin with the following lemmata, which establishes the the asymptotic normality of the the potential function for a sequence of graphs satisfying two technical conditions.
In this section we will use the reparameterization as an Ising game.

\vspace{5pt}

\begin{lemma}[Martingale Central Limit Theorem {\cite{HallHeyde1980}}]\label{lem:MCLT}
For a martingale difference sequence $\{D_k\}_{k=1}^N$ with a filtration $\{\mathcal{F}_k\}_{k=1}^N$, define $V_N \Equaldef \sum_{k=1}^N \mathbb{E}[D_k^2 \mid \mathcal{F}_{k-1}]$
and $\sigma_N^2 \Equaldef \sum_{k=1}^N \mathbb{E}[D_k^2]$. If
\begin{enumerate}[(i)]
  \item $\max_{1 \le k \le N} |D_k| / \sigma_N \xrightarrow{P} 0$,
  \item $V_N / \sigma_N^2 \xrightarrow{P} 1$,

\end{enumerate}
then
\begin{equation}
  \frac{\sum_{k=1}^N D_k}{\sigma_N}
  \;\xrightarrow{d}\;
  \mathcal{N}(0,1).
\end{equation}
\end{lemma}

\vspace{5pt}

\begin{figure*}[t!]
    \centering
    \includegraphics[width=0.99\textwidth]{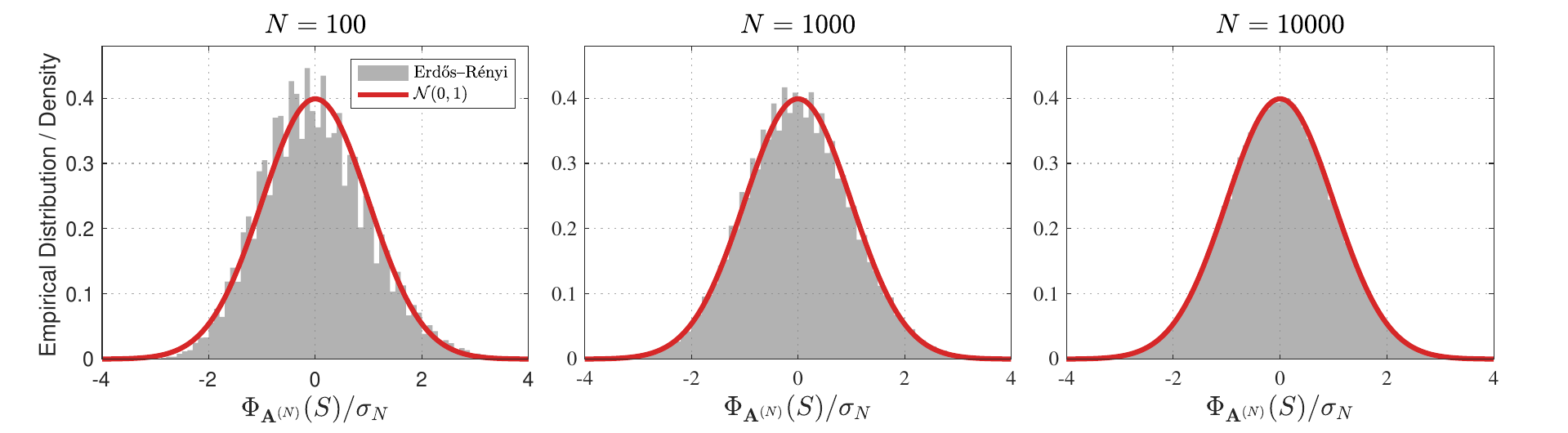}
    \caption{Convergence in distribution of the normalized potential function $\Phi_N(S)/\sigma_N$ to a standard Gaussian for Erd\"{o}s--R\'enyi random graphs $G(N, p)$ with $p = 10/N$ and $\theta = 0.3$. As $N$ grows, the empirical density (histogram) concentrates around the $\mathcal{N}(0,1)$ density (solid curve), illustrating \Cref{thm:CLT} for the sparse graph regime.}
    \label{fig:MCLT}
\end{figure*}

\vspace{5pt}

\begin{lemma}\label{thm:CLT}
Let $\{\mathbf{A}^{(N)}\}$ be a sequence of adjacency matrices of simple undirected graphs
with degrees $d_i^{(N)}$ and $|\mathcal{E}_N|$ edges, and let
$S=(S_1,\dots,S_N)$ be an i.i.d. sequence of Rademacher random variables.  Define
\begin{equation}
\tilde{\Phi}_N(S)
\Equaldef
\frac18\, S^\mathsf{T}\mathbf{A}^{(N)} S
+
\Big(\frac14 - \frac{\theta}{2}\Big)\,\mathbb{1}^\mathsf{T}\mathbf{A}^{(N)} S,
\end{equation}
and set $c \Equaldef \frac14 - \frac{\theta}{2}$ and
\begin{equation}\label{eq:variance_potential}
\sigma_N^2 \Equaldef
\frac{|\mathcal{E}_N|}{16}
+
c^2\sum_{i=1}^N \big(d_i^{(N)}\big)^2.
\end{equation}

Assume that $\mathbf{A}^{(N)}$ satisfies the following conditions:
\begin{enumerate}[(i)]
\item $\Delta_N \Equaldef \max_i d_i^{(N)} = o(\sigma_N)$,
\item $\sum_{i=1}^N (d_i^{(N)})^2 = O(\sigma_N^2)$.
\end{enumerate}
Then,
\begin{equation}
\frac{\tilde{\Phi}_N(S) - \mathbb{E}\big[\tilde{\Phi}_N(S)\big]}{\sigma_N}
\;\xlongrightarrow{D}\;
\mathcal{N}(0,1).
\end{equation}
\end{lemma}

\vspace{5pt}

\begin{IEEEproof} We apply the martingale central limit theorem in \cref{lem:MCLT}. We begin by writing the potential function as
\begin{equation}
\tilde{\Phi}_N(S)
= \frac14\sum_{1\le i<j\le N} A^{(N)}_{ij}\, S_i S_j +
c \sum_{i=1}^N d_i^{(N)} S_i.
\end{equation}
Since $\mathbb{E}[S_i]=0$ and $\mathbb{E}[S_i S_j]=0$ for $i\neq j$, we have $\mathbb{E}\big[\tilde{\Phi}_N(S)\big]=0$ and
\begin{equation}
\mathrm{Var}\big(\tilde{\Phi}_N(S)\big)
=
\frac{|\mathcal{E}_N|}{16}
+
c^2\sum_{i=1}^N (d_i^{(N)})^2 \Equaldef \sigma^2_N.
\end{equation}

Define the following filtration\footnote{Here $\sigma(S_1,\dots,S_k)$ denotes the smallest $\sigma$-algebra generated by $S_1,\ldots,S_k$. Not to be confused with standard deviation.}
\begin{equation}
\mathcal{F}_k \Equaldef \sigma(S_1,\dots,S_k)
\end{equation}
and 
\begin{equation}
M_k \Equaldef \mathbb{E}\big[\tilde{\Phi}_N(S)\mid \mathcal{F}_k\big],
\qquad
D_k \Equaldef M_k - M_{k-1}.
\end{equation}
Since $\mathbb{E}\big[S_{\ell} \mid \mathcal{F}_k\big]=0$ for $\ell>k$, we have
\begin{equation}
M_k
=
\frac14 \sum_{\substack{1\le i<j\le k}} A^{(N)}_{ij}\, S_i S_j
+
c \sum_{i=1}^k d_i^{(N)} S_i,
\end{equation}
and therefore
\begin{equation}
D_k
=
S_k\!\left(\frac14\sum_{i=1}^{k-1} A^{(N)}_{ik}\, S_i + c\, d_k^{(N)}\right).
\end{equation}
Observe that $D_k$ is $\mathcal{F}_k$-measurable and
\begin{align}
\mathbb{E}[D_k\mid \mathcal{F}_{k-1}] &= \mathbb{E}[S_k \mid \mathcal{F}_{k-1}]\left(\frac14\sum_{i=1}^{k-1} A^{(N)}_{ik}\, S_i + c\, d_k^{(N)}\right)\\ &\overset{(a)}{=} 0,
\end{align}
where $(a)$ follows from $S_k$ being independent of $\mathcal{F}_{k-1}$ and having zero mean. Therefore, $\{D_k\}_{k=1}^N$ is a martingale difference sequence \cite{Cinlar2011}, and
\begin{equation}
\tilde{\Phi}_N(S) = \sum_{k=1}^N D_k.
\end{equation}

Since $|S_i|=1$, we can bound the increments using the triangle inequality as
\begin{equation}
|D_k|
\;\le\;
\Big(\frac{1}{4} + |c|\Big)\,d_k^{(N)}.
\end{equation}
By assumption~(i),
\begin{equation}
\max_{1\le k\le N} \frac{|D_k|}{\sigma_N}
\;\le\;
\Big(\frac{1}{4} + |c|\Big)\,\frac{\Delta_N}{\sigma_N}
 \longrightarrow 0,
\end{equation}
which verifies condition~(i) of \cref{lem:MCLT}.

Next, consider the conditional variance process defined as
\begin{equation}\label{eq:V_N_1}
V_N
\;\Equaldef\;
\sum_{k=1}^N \mathbb{E}\!\big[D_k^2\mid \mathcal{F}_{k-1}\big].
\end{equation}
Since $S_k^2 = 1$ and $S_k$ is independent of $\mathcal{F}_{k-1}$, we have
\begin{equation}\label{eq:V_N_2}
\mathbb{E}\big[D_k^2\mid\mathcal{F}_{k-1}\big]
=
\left(
\frac{1}{4}\sum_{i=1}^{k-1} A^{(N)}_{ik}\, S_i + c\, d_k^{(N)}
\right)^2,
\end{equation}
and taking its expectation, we obtain
\begin{equation} \label{eq:V_N_mean}
\mathbb{E}[V_N]
=
\sum_{k=1}^N \mathbb{E}[D_k^2]
=
\sigma_N^2.
\end{equation}

It remains to verify condition~(ii) of \cref{lem:MCLT}.
Since $\mathbb{E}[V_N]=\sigma_N^2$, we will show that $\mathrm{Var}(V_N)=o(\sigma_N^4)$.

Now consider the from $S$ we construct new sequence of random variables 
\begin{equation}
S^{(i)}=(S_1,\ldots,S_{i-1},S_i',S_{i+1},\ldots,S_N),
\end{equation}
where $S_i'$ is an independent copy of $S_i$ with the same distribution.
Let $V_N^{(i)}$ denote the same quantity as $V_N$ constructed from $S^{(i)}$ instead of $S$. From \eqref{eq:V_N_2}, the $k$-th term in $V_N$ given in \eqref{eq:V_N_1} depends on $S_i$ only if $k > i$ and $A_{ik}^{(N)} = 1$.  Since $S_i' - S_i \in \{-2, 0, 2\}$ and enters
the expression inside the square of \eqref{eq:V_N_2} linearly, the change in the $k$-th summand is at most $C\, d_k^{(N)}$ for a constant~$C$ that depends only on~$c$.  Summing over the neighbors of node~$i$ such that $k > i$ and bounding $d_k^{(N)} \le \Delta_N$, we obtain
\begin{equation}
  \big|V_N - V_N^{(i)}\big|
  \;\le\;
  C\!\sum_{\substack{k>i \\ A_{ik}^{(N)}=1}} d_k^{(N)}
  \;\le\;
  C\, d_i^{(N)}\,\Delta_N.
\end{equation}
By the Efron--Stein inequality (cf.\ \Cref{lem:EF} in Appendix \ref{sec:lemmata}),
\begin{equation}\label{eq:variance_bound}
  \operatorname{Var}(V_N)
  \;\le\;{}
  \frac{1}{2}\sum_{i=1}^{N}
    \mathbb{E}\!\Big[\big(V_N - V_N^{(i)}\big)^2\Big]
  \;\le\;
  \frac{C^2}{2}\,\Delta_N^2\sum_{i=1}^{N}
    \big(d_i^{(N)}\big)^2.
\end{equation}

Since $\mathbb{E}[V_N] =  \sigma_N^2$,
using \eqref{eq:V_N_mean}, we have 
\begin{equation}
  \frac{V_N}{\sigma_N^2} - 1
  \;=\;
  \frac{V_N - \mathbb{E}[V_N]}{\sigma_N^2}.
\end{equation}
By Chebyshev's inequality and the variance bound in \eqref{eq:variance_bound},
\begin{equation}
  \mathbb{P}\!\bigg(\bigg|\frac{V_N}{\sigma_N^2} - 1\bigg| > \varepsilon\bigg)
  \;\le\;
  \frac{\operatorname{Var}(V_N)}{\varepsilon^2\,\sigma_N^4}
  \;\le\;
  \frac{C^2}{2\varepsilon^2}\,
  \frac{\Delta_N^2\sum_{i=1}^{N} \big(d_i^{(N)}\big)^2}{\sigma_N^4}.
\end{equation}
By assumptions~(i) and~(ii), $\Delta_N^2/\sigma_N^2 \to 0$ and $\sum_{i=1}^{N}(d_i^{(N)})^2 = O(\sigma_N^2)$, we have
\begin{equation}
  \frac{\Delta_N^2\sum_{i=1}^{N} \big(d_i^{(N)}\big)^2}{\sigma_N^4}
  \;\longrightarrow\; 0.
\end{equation}
Therefore, $\mathrm{Var}(V_N) = o(\sigma_N^4)$ which implies
$V_N/\sigma_N^2 \xrightarrow{P} 1$.

Finally, from \cref{lem:MCLT}, we have 
\begin{equation}
\frac{\Phi_N(S)-\mathbb{E}\big[\Phi_N(S)\big]}{\sigma_N} \overset{D}\longrightarrow \mathcal{N}(0,1).
\end{equation}
\end{IEEEproof}

\vspace{5pt}

The Gaussian approximation established in \cref{thm:CLT} provides a way to establish the optimality of regular graphs in the asymptotic regime $N\to\infty$.

\vspace{5pt}

\begin{theorem}\label{thm:regular_optimal}
Consider the ensemble of adjacency matrices $\mathbf{A}^{(N)}$ that satisfy the following asymptotic conditions:
\begin{enumerate}[(i)]
\item $\Delta_N \Equaldef \max_i d_i^{(N)} = o(\sigma_N)$,
\item $\sum_{i=1}^N (d_i^{(N)})^2 = O(\sigma_N^2)$.
\end{enumerate}
Then, for sufficiently large $N$ and for all $\beta > 0$, the stationary probability 
$\tilde{\mu}_{\mathbf{A}^{(N)}}(a^\star \mid \beta)$ is maximized, over all graphs 
on $N$ vertices with $|\mathcal{E}_N|$ edges, by the $K$-regular graph 
with $K = 2|\mathcal{E}_N|/N$, or by a near-$K$-regular graph when 
$2|\mathcal{E}_N|/N$ is not an integer.
\end{theorem}

\vspace{5pt}

\begin{IEEEproof} Under conditions (i) and (ii), \Cref{thm:CLT} implies that a Gaussian approximation for $\tilde{\Phi}_{\mathbf{A}^{(N)}}(S)$ holds. That is, $\tilde{\Phi}_{\mathbf{A}^{(N)}}(S) \approx \mathcal{N}(0,\sigma_N^2)$, and thus the MGF in the denominator of the objective function is approximately
\begin{equation}\label{eq:gauss_mgf}
\mathbb{E}_S\!\left[e^{\beta\tilde{\Phi}_{\mathbf{A}^{(N)}}(S)}\right]
\approx
e^{\beta^2\sigma_N^2/2}.
\end{equation}
Since the exponential function is monotone increasing, minimizing \eqref{eq:gauss_mgf} over the class of graphs for a given $|\mathcal{E}_N|$ edges reduces to minimizing $\sigma_N^2$. Recall from \eqref{eq:variance_potential} that
\begin{equation}\label{eq:sigma_degree}
\sigma_N^2
= \frac{|\mathcal{E}_N|}{16}
+ c^2\sum_{i=1}^{N} \big( d_i^{(N)}\big)^2.
\end{equation}
The first term is fixed for a given $|\mathcal{E}_N|$.  Hence, the network design problem reduces to
\begin{equation}\label{eq:degree_opt}
\begin{aligned}
& \underset{d_1^{(N)},\ldots,d^{(N)}_N}{\mathrm{minimize}}
& & \sum_{i=1}^N \big(d^{(N)}_i\big)^2 \\
& \text{subject to}
& & \sum_{i=1}^N d_i^{(N)} = 2|\mathcal{E}_N|, \quad d_i^{(N)} \in \mathbb{Z}.
\end{aligned}
\end{equation}
We can solve this optimization problem and characterize the optimal degree distribution using Majorization theory as in \eqref{eq:majorization1} and \eqref{eq:majorization2}. Therefore, the optimal degree sequence is regular or near-regular as in \cref{thm:small_beta} depending on the prescribed number of nodes $N$ and edges $\mathcal{E}_N$ in $\mathbf{A}_N$.
\end{IEEEproof}

\vspace{5pt}

\subsection{Price of Irregularity - PoI}
\cref{thm:small_beta,thm:moderate_beta,thm:regular_optimal} provide a new design principle for multi-agent networks. Given a connectivity budget of $|\mathcal{E}_N|$ communication links among $N$ homogeneously bounded rational agents, the system designer should distribute edges as evenly as possible. The detailed topology matters for intermediate values of rationality, but when $\beta\rightarrow 0$ and when $N\rightarrow \infty$ (under mild technical conditions), only the degree distribution matters, and the optimal graphs are either regular or near-regular.

In the limit $N \to \infty$, the {\it ``price of irregularity''} is captured by the empirical degree variance $\mathrm{Var}(d^{(N)})$, where $d^{(N)}$ denotes the degree distribution of $\mathbf{A}_{N}$, i.e., $d^{(N)}=\mathbf{1}^\top\mathbf{A}_{N}$. To see this, compare a $K$-regular graph $\mathbf{A}_N'$ with an arbitrary graph $\mathbf{A}_N''$ having the same number of edges. Since both graphs share the same $|\mathcal{E}_N|$, the potential functions evaluated at the coordinated action profile $a^\star\in \{\mathbb{1}, \mathbb{0}\}$ coincide, i.e.\ $\Phi_{\mathbf{A}_N'}(a^\star) = \Phi_{\mathbf{A}_N''}(a^\star)$, therefore
\begin{equation}
\mathrm{PoI} \Equaldef \log \frac{\mu_{\mathbf{A}_N'}(a^\star)}{\mu_{\mathbf{A}_N''}(a^\star)} = \log Z_{\mathbf{A}_N'} - \log Z_{\mathbf{A}_N''}.
\end{equation}
Applying the Gaussian approximation from \cref{thm:CLT}, each log-partition function is determined by the variance of the potential:
\begin{equation}
\log Z_{\mathbf{A}_N} \approx N\log 2 + \frac{\beta^2}{2}\sigma_N^2,
\end{equation}
where
\begin{equation}
\sigma_N^2 = \frac{|\mathcal{E}_N|}{16} + c^2 \sum_{j=1}^N d_j^2, \qquad c = \frac{1}{4} - \frac{\theta}{2}.
\end{equation}
Since the term $\frac{|\mathcal{E}_N|}{16}$ is common to both graphs, the difference only depends on the degree distribution. We decompose
\begin{equation}
\sum_{j=1}^N d_j^2 = N\!\left(\bar{d}^2 + \mathrm{Var}(d)\right),
\end{equation}
and note that the average degree of any graph is
\begin{equation}
\bar{d} = \frac{1}{N}\sum_{j=1}^N d_j = \frac{2|\mathcal{E}_N|}{N}.
\end{equation}
Since both graphs share the same $|\mathcal{E}_N|$ and $N$, they have the same average degree $\bar{d} = K$. Since $\mathrm{Var}(d') = 0$ for the $K$-regular graph $\mathbf{A}_N'$, we obtain
\begin{equation}\label{eq:PoI}
\mathrm{PoI} = \log \frac{\mu_{\mathbf{A}_N'}(a^\star)}{\mu_{\mathbf{A}_N''}(a^\star)} \approx \frac{\beta^2 c^2 N}{2}\,\mathrm{Var}(d'')>0,
\end{equation}
where $d''$ is the degree distribution of $\mathbf{A}_N''$. This penalty is larger for networks with heterogeneous degree distributions and large $N$. It scales quadratically in both the rationality parameter $\beta$ and the parameter $c = \frac{1}{4} - \frac{\theta}{2}$, vanishing only at $\theta = \frac{1}{2}$, where the potential becomes insensitive to degree heterogeneity.

To empirically validate the price of irregularity, we generate random connected graphs with the same number of edges as a $K$-regular graph ($K=6$) across several values of $N$ and compare the Gibbs probability of the optimal action profile on each irregular graph to that of the regular graph. Figure~\ref{fig:cost_of_irregularity} plots the log-ratio $\log(\mu_{\mathbf{A}_N'}/\mu_{\mathbf{A}_N''})$ against $N \cdot \mathrm{Var}(d)$ for each irregular graph. The results confirm an approximately linear relationship between the $\mathrm{PoI}$ and $N\,\mathrm{Var}(d)$, which holds even for modest values of $N$. This is consistent with the asymptotic expression obtained in \eqref{eq:PoI}.

\vspace{5pt}

\begin{figure}[h!]
    \centering
\includegraphics[width=0.99\linewidth]{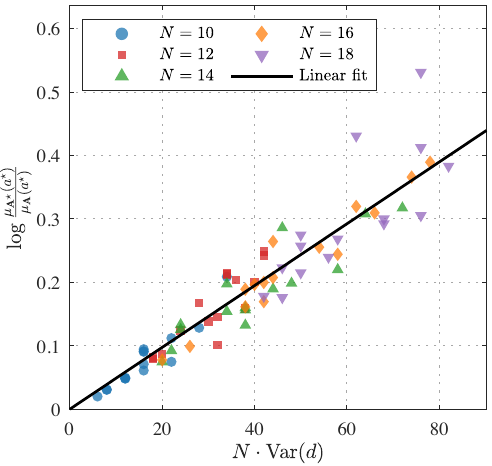}
    \caption{Price of Irregularity}
\label{fig:cost_of_irregularity}
\end{figure}

\section{Conclusions and Future Work}\label{sec:conclusions}
We studied the problem of learning to coordinate with bounded rationality over a network. Assuming the agents adhere to LLL with a homogeneous bounded rationality parameter, we showed that the stationary probability of coordinated (risk-dominant) action profiles can be increased by improving connectivity in regular graphs. We also showed that more connected networks can operate with agents with lower rationality and still achieve a given level of coordination than less connected ones. This is the first design principle from this work. This creates a form of \textit{Wisdom of Crowds} \cite{Becker:2017}, where a NE (approximately) emerges by aggregating information from neighbors and decisions propagating imperfectly over the network.

Additionally, we proved that when the networks are irregular, the stationary probability of a coordinated action profile is monotone increasing with respect to the operation of adding new edges to the graph. Finally, we proved that for small rationality, and for networks with a sufficiently large number of agents, regular (or near-regular) graphs are optimal when the agents have homogeneous bounded rationality. When rationality and the number of agents is moderate, regular graphs maximize a lower bound on the stationary probability and are a robust choice for maximizing the coordination of bounded rational agents. This is the second design principle from this work.

This work can be extended in many different directions. The first is to consider the possibility of having agents with heterogeneous bounded rationalities, and design the connectivity among them so as to promote coordination. In particular, we are interested in the question of whether higher rationality agents must be more connected among themselves (segregation), or if connectivity should be distributed by connecting lower rationality agents to agents with higher rationality. Yet another important research problem is to consider the scheduling of agents in heterogeneous systems. In that case, we would like to determine the optimal agent schedule to maximize their ultimate coordination. We would like to determine whether higher rationality agents should be updated more or less frequently than other agents. Finally, we suggest the generalization of this approach to handle very large scale graphs using graphons \cite{Parise:2021}.

\section{Acknowledgments}
The authors would like to thank Dr.~Yifei Zhang for discussions and 
insightful suggestions at an early stage of this work.

\appendices

\section{Auxiliary Results}\label{sec:lemmata}

\begin{lemma}[Efron--Stein inequality]\label{lem:EF}
Let $X_1,\dots,X_n$ be independent random variables and let
$f(X_1,\dots,X_n)$ be square-integrable.  For each $1\le i\le n$,
let $X_i'$ be an independent copy of $X_i$ and define
\begin{equation}
  f^{(i)} = f(X_1,\dots,X_{i-1},X_i',X_{i+1},\dots,X_n).
\end{equation}
Then
\begin{equation}
  \mathrm{Var}(f)
  \;\le\;
  \frac{1}{2}\sum_{i=1}^{n}
  \mathbb{E}\!\big[(f - f^{(i)})^{2}\big].
\end{equation}
\end{lemma}

\section{Proofs}\label{sec:proofs}

\subsection{Proof of Theorem 1}\label{sec:proof1}

Consider the following potential function
\begin{equation}\label{potential_f}
    \Phi(a) \Equaldef \frac{1}{2}\sum\limits_{i\in [N]} \sum\limits_{j\in \mathcal{N}_{i}} \phi(a_{i},a_{j}),
\end{equation}
where
\begin{equation}\label{eq:potential_two_agents}
\phi (a_i,a_j) \Equaldef a_ia_j + (1-a_i-a_j)\theta.
\end{equation}
The function $\phi$ is an exact potential for the two-player game with payoff $V_i$. Therefore, the following holds:
\begin{equation}
\phi(a_i',a_j) - \phi(a_i'',a_j) = V_i(a_i',a_j) - V_i(a_i'',a_j),
\end{equation}
for all $a_i',a_i''\in\{0,1\}$ such that $a_i'\neq a_i''$. We proceed by verifying that the function in \eqref{potential_f} satisfies the condition in \eqref{eq:exact_potential1}.
Let $m \in [N]$, and $a'_m,a''_m \in\{0,1\}$ such that $a'_m\neq a''_m$. Then,
\begin{multline}\label{globalpotental_minus}
     \Phi(a'_{m},a_{-m})-\Phi(a''_{m},a_{-m}) = \\ \frac{1}{2}\sum_{i\in[N]}\sum_{j\in\mathcal{N}_i} \phi(a_i,a_j)\Bigg|_{(a_m',a_{-m})}  \\ - \frac{1}{2}\sum_{i\in[N]}\sum_{j\in\mathcal{N}_i} \phi(a_i,a_j)\Bigg|_{(a_m'',a_{-m})}.
\end{multline}
Then, notice that
\begin{equation}
\sum_{i\in[N]}\sum_{j\in\mathcal{N}_i} \phi(a_i,a_j) = \sum_{j\in\mathcal{N}_m} \phi(a_m,a_j) + \sum_{i\neq m}\sum_{j\in\mathcal{N}_i} \phi(a_i,a_j).
\end{equation}
Recall that
\begin{equation}
\phi(a_m',a_j) - \phi(a_m'',a_j) = V_m(a_m',a_j) - V_m(a_m'',a_j).
\end{equation}
Therefore,
\begin{multline}\label{eq:intermediate_step_potential1}
     \Phi(a'_{m},a_{-m})-\Phi(a''_{m},a_{-m}) =  \\ \frac{1}{2}\sum_{j\in\mathcal{N}_m} \big[ V_m(a_m',a_j) - V_m(a_m'',a_j) \big] \\ +
      \frac{1}{2}\sum_{i\neq m} \sum_{j\in\mathcal{N}_i} \phi(a_i,a_j)\Bigg|_{(a_m',a_{-m})} \\ -   \frac{1}{2} \sum_{i\neq m} \sum_{j\in\mathcal{N}_i} \phi(a_i,a_j)\Bigg|_{(a_m'',a_{-m})}.
\end{multline}
The first term in \eqref{eq:intermediate_step_potential1} is equal to
\begin{equation}\label{eq:half_difference_potential}
\frac{1}{2}\big[U_m(a'_m,a_{-m}) - U_m(a''_m,a_{-m})\big].
\end{equation}
We proceed with showing that the remaining two terms yield an identical contribution. 
For all $i\neq m$ such that $m\notin \mathcal{N}_i$,
\begin{equation}
\sum_{j\in\mathcal{N}_i} \phi(a_i,a_j)\Bigg|_{(a_m',a_{-m})} = \sum_{j\in\mathcal{N}_i} \phi(a_i,a_j)\Bigg|_{(a_m'',a_{-m})}.
\end{equation}
Define the following set
\begin{equation}
\mathcal{S}_m \Equaldef\big\{i \mid  i \neq m \textup{ and } m\in\mathcal{N}_i \big\}
\end{equation}
and evaluate the difference
\begin{multline}
\sum_{i \in \mathcal{S}_m} \sum_{j\in \mathcal{N}_i} \phi(a_i,a_j) \Bigg|_{(a_m',a_{-m})} \\ - \sum_{i \in \mathcal{S}_m} \sum_{j\in \mathcal{N}_i} \phi(a_i,a_j) \Bigg|_{(a_m'',a_{-m})},
\end{multline}
which is equal to
\begin{multline}\label{eq:difference1}
\sum_{i \in \mathcal{S}_m}  \phi(a_i,a'_m) \Bigg|_{(a_m',a_{-m})} \! \! \! - \sum_{i \in \mathcal{S}_m}  \phi(a_i,a_m'') \Bigg|_{(a_m'',a_{-m})}.
\end{multline}
Since
\begin{equation}
\phi(a_i,a_j) = \phi(a_j,a_i),
\end{equation}
we have that \eqref{eq:difference1} is equal to
\begin{equation}
\sum_{i\in\mathcal{S}_m} \big[\phi(a_m',a_i) - \phi(a_m'',a_i)\big].
\end{equation}
Finally, since $\phi$ is a potential function for the two-player game between $i$ and $j$ when $j \in \mathcal{N}_i$, and the graph is undirected so that $\mathcal{S}_m = \mathcal{N}_m$, we have
\begin{multline}
\sum_{i\in\mathcal{S}_m} \big[V_m(a_m',a_i) - V_m(a_m'',a_i)\big] \\ = U_m(a_m',a_{-m}) - U_m(a_m'',a_{-m}).
\end{multline}
Combining the two contributions, we obtain
\begin{equation}
\Phi(a'_m,a_{-m}) - \Phi(a''_m,a_{-m}) = U_m(a'_m,a_{-m}) - U_m(a''_m,a_{-m}),
\end{equation}
which is the condition in \eqref{eq:exact_potential1}.
\hfill $\blacksquare$


\bibliography{ref,ref_new_entries-2}
\bibliographystyle{ieeetr}

\end{document}